\documentclass[11pt]{article}

\usepackage{latexsym}
\usepackage{multirow}
\usepackage{amsmath}
\usepackage{amssymb}
\usepackage{amsfonts,amsthm}
\usepackage{adjustbox}
\usepackage{mathcomp}

\usepackage[top=1in, bottom=1in, left=1in, right=1in]{geometry}

\usepackage{tabu}
\usepackage{longtable}

\usepackage{tikz}
\usetikzlibrary{arrows,decorations.pathmorphing,decorations.shapes,backgrounds,fit}
\pgfkeys{tikz/cc1/.style={dotted} }
\pgfkeys{tikz/cc4/.style={dashed} }
\pgfkeys{tikz/cc3/.style={decorate,decoration=bumps} }
\pgfkeys{tikz/cc2/.style={decorate,decoration=snake} }
\pgfkeys{tikz/solid/.style={line width=2pt} }
\pgfkeys{tikz/solid2/.style={line width=1pt} }
\pgfkeys{tikz/nodostile1/.style={draw,line width=1pt}}
\pgfkeys{tikz/nodostile2/.style={draw,line width=2pt}}
\pgfkeys{tikz/arcostile1/.style={very thick} }
\pgfkeys{tikz/arcostile2/.style={red,very thick} }
\pgfkeys{tikz/arcostile3/.style={thick} }
\pgfkeys{tikz/arcostile4/.style={line width=2pt} }

\newcommand{\ignore}[1]{}

\newtheorem{theorem}{Theorem}[section]
\newtheorem{lemma}[theorem]{Lemma}

\newtheorem{corollary}[theorem]{Corollary}

\usepackage{dcolumn}
\newcolumntype{d}[1]{D{.}{.}{#1}}

\usepackage{array}
\newcolumntype{C}[1]{>{\centering\let\newline\\\arraybackslash\hspace{0pt}}m{#1}}

\newcommand{\sidenote}[1]{\textbf{(*)}\marginpar {\tiny \raggedright{(*) #1}}}

\begin{document}
\title{\bf Approximating the Smallest Spanning Subgraph for 2-Edge-Connectivity in Directed Graphs\thanks{
A preliminary version of some of the results of this work was presented at ESA 2015.}}
\author{
Loukas Georgiadis\thanks{
University of Ioannina, Greece. E-mails: \texttt{\{loukas,charis,nparotsi\}@cs.uoi.gr}.}
\and
Giuseppe F. Italiano\thanks{
Universit\`a di Roma ``Tor Vergata'', Italy. E-mail: \texttt{giuseppe.italiano@uniroma2.it}.
Partially supported by MIUR
under Project AMANDA.
}
\and
Charis Papadopoulos$^\dagger$
\and
Nikos Parotsidis$^\dagger$
}

\date{\today}

\maketitle

\begin{abstract}
Let $G$ be a strongly connected directed graph. We consider the following three problems, where we wish to compute the smallest strongly connected spanning subgraph of $G$ that maintains respectively: the $2$-edge-connected blocks of $G$ (\textsf{2EC-B});
the $2$-edge-connected components of $G$ (\textsf{2EC-C});
both the $2$-edge-connected blocks and the $2$-edge-connected components of $G$ (\textsf{2EC-B-C}). All three problems are NP-hard, and thus we are interested in efficient approximation algorithms.
For \textsf{2EC-C} we can obtain a $3/2$-approximation by combining previously known results.
For \textsf{2EC-B} and \textsf{2EC-B-C}, we present new $4$-approximation algorithms that run in linear time.
We also propose various heuristics to improve the size of the computed subgraphs in practice, and conduct a thorough experimental study to assess their merits in practical scenarios.
\end{abstract}

\section{Introduction}
\label{sec:introduction}

Let $G=(V,E)$ be a directed graph (digraph), with $m$ edges and $n$ vertices.
An edge of $G$ is a \emph{strong bridge} if its removal increases the number of strongly connected components of $G$.
A digraph $G$ is $2$-edge-connected if it has no strong bridges. The $2$-edge-connected components of $G$ are its maximal $2$-edge-connected subgraphs.
Let $v$ and $w$ be two distinct vertices: $v$ and $w$ are \emph{$2$-edge-connected}, denoted by  $v \leftrightarrow_{\mathrm{2e}} w$, if there are two edge-disjoint directed paths from $v$ to $w$  and two edge-disjoint directed paths from $w$ to $v$. (Note that a path from $v$ to $w$ and a path from $w$ to $v$ need not be edge-disjoint.)
A \emph{$2$-edge-connected block} of $G=(V,E)$ is a maximal subset $B \subseteq V$ such that $u \leftrightarrow_{\mathrm{2e}} v$ for all $u, v \in B$.
Differently from undirected graphs, in digraphs $2$-edge-connected blocks can be different from the  $2$-edge-connected components, i.e., two vertices may be $2$-edge-connected but lie in different $2$-edge-connected components. See Figure \ref{fig:2ebc}.

A \emph{spanning subgraph} $G'$ of $G$ has the same vertices as $G$ and contains a subset of the edges of $G$.
Computing a smallest
spanning subgraph (i.e., one with minimum number of edges)
that maintains the same edge or vertex connectivity properties of the original graph is a fundamental problem in network design, with many practical applications~\cite{connectivity:nagamochi-ibaraki}.
In this paper we consider the problem of finding the smallest spanning subgraph of $G$ that
maintains  certain $2$-edge-connectivity requirements in addition to strong connectivity.
Specifically, we distinguish three 
problems that we refer to as \textsf{2EC-B}, \textsf{2EC-C} and \textsf{2EC-B-C}. In particular, we wish to compute the smallest strongly connected spanning subgraph of a digraph $G$ that maintains the following properties:
\begin{itemize}
\item[(i)] the pairwise $2$-edge-connectivity of $G$, i.e., the $2$-edge-connected blocks of $G$ (\textsf{2EC-B});
\item[(ii)] the $2$-edge-connected components of $G$ (\textsf{2EC-C});
\item[(iii)] both the $2$-edge-connected blocks and the $2$-edge-connected components of $G$ (\textsf{2EC-B-C}).
\end{itemize}
Since all those problems are NP-hard~\cite{GJ:NP},  we are interested in designing efficient approximation algorithms.

While for \textsf{2EC-C} one can obtain a $3/2$-approximation using known results, for the other two problems no efficient approximation algorithms were previously known.
Here we present a linear-time algorithm for \textsf{2EC-B} that achieve an approximation ratio of $4$.
Then we extend this algorithm so that it approximates the smallest \textsf{2EC-B-C} also within a factor of $4$.
This algorithm runs in linear time if the $2$-edge-connected components of $G$ are known, otherwise it requires the computation of these components, which can be done in $O(n^2)$ time \cite{2CC:HenzingerKL14}.
Moreover, we give efficient implementations of our algorithms that run very fast in practice.
Then we consider various heuristics that improve the size of the computed subgraph in practice.
Some of these heuristics require $O(mn)$ time in the worst case, so we also consider various techniques that achieve significant speed up.

\begin{figure}[t]
\begin{center}
\includegraphics[trim={0 0 0 10cm}, clip=true, width=\textwidth]{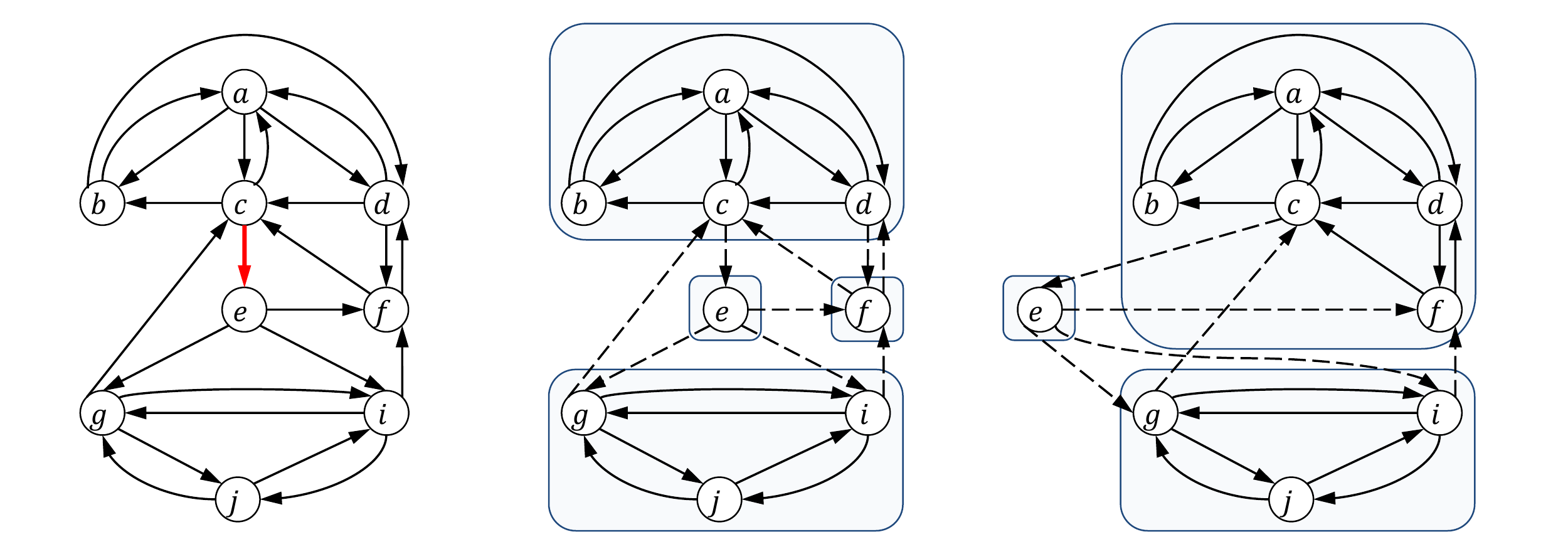}
\end{center}
\vspace{-0.2cm}
\hspace{2cm} (a) $G$ \hspace{3.2cm} (b) $2\mathit{ECC}(G)$ \hspace{3.2cm}  (c) $2\mathit{ECB}(G)$
\caption{(a) A strongly connected digraph $G$, with a strong bridge shown in red (better viewed in color).
(b) The $2$-edge-connected components and 
(c)
the $2$-edge-connected
blocks of $G$.}
\label{fig:2ebc}
\end{figure}

\subsection{Related work}
\label{sec:related-work}
Finding a smallest $k$-edge-connected (resp. $k$-vertex-connected) spanning subgraph of a given $k$-edge-connected (resp. $k$-vertex-connected) digraph is NP-hard for $k \ge 2$ for undirected graphs, and for $k \ge 1$ for digraphs~\cite{GJ:NP}.
More precisely, if the input graph consists of a single $2$-edge-connected block then the problem asks for the smallest $2$-edge-connected subgraph, whereas
if the input graph consists of $n$ singleton $2$-edge-connected blocks then the problem coincides with the smallest strongly connected spanning subgraph.
Problems of this type, together with more general variants of approximating minimum-cost subgraphs that satisfy certain connectivity requirements, have received a lot of attention, and several important results have been obtained. 
More general problems of approximating minimum-cost subgraphs that satisfy certain connectivity requirements has also received a lot of attention; see, e.g., the survey \cite{MCC:KZ:2007}.

Currently, the best approximation ratio for computing the smallest strongly connected spanning subgraph (\textsf{SCSS}) is $3/2$, achieved by Vetta with a polynomial-time algorithm~\cite{Vetta:MSCS:2001}. 
Although Vetta did not analyze exactly the running time of his algorithm, it needs to solve a maximum matching problem in a relaxation problem. 
A faster linear-time algorithm that achieves a $5/3$-approximation was given by Zhao et al.~\cite{ZNI:MSCS:2003}.
For the smallest $k$-edge-connected spanning subgraph (\textsf{kECSS}), Laehanukit et al.~\cite{LGS:MSCS:2012} gave a randomized $(1+1/k)$-approximation algorithm.
Regarding hardness of approximation, Gabow et al. \cite{GGTW:kECSS:2009} showed that there exists an
absolute constant $c>0$ such that for any integer $k \ge 1$, approximating the smallest \textsf{kECSS} on directed multigraphs to
within a factor $1 + c/k$ in polynomial time implies $\mathrm{P} = \mathrm{NP}$. Jaberi~\cite{2vcb:jaberi} considered various optimization problems related to \textsf{2EC-B} and proposed corresponding approximation algorithms. The approximation ratio in Jaberi's algorithms, however, is linear in the number of strong bridges, and hence $O(n)$ in the worst case.

\subsection{Our results}
\label{sec:our-results}
In this paper we provide both theoretical and experimental contributions to the \textsf{2EC-B}, \textsf{2EC-C} and \textsf{2EC-B-C} problems.
%
A $3/2$-approximation for \textsf{2EC-C} can be obtained by carefully combining the \textsf{2ECSS} randomized algorithm of Laehanukit et al.~\cite{LGS:MSCS:2012} and the \textsf{SCSS} algorithm of Vetta~\cite{Vetta:MSCS:2001}. A faster and deterministic 2-approximation algorithm for \textsf{2EC-C} can be obtained by combining techniques based on edge-disjoint spanning trees~\cite{edge-disjoint:edmonds,st:t} with the \textsf{SCSS} algorithm of  Zhao et al.~\cite{ZNI:MSCS:2003}.
We remark that the other two problems considered here, \textsf{2EC-B} and \textsf{2EC-B-C}, seem harder to approximate.
The only known result is the \emph{sparse certificate} for $2$-edge-connected blocks of \cite{2ECB}.
In this context, a sparse certificate $C(G)$ of a strongly connected digraph $G$ is a spanning subgraph of $G$ with $O(n)$ edges.
Such a sparse spanning subgraph implies a linear-time $O(1)$-approximation algorithm for \textsf{2EC-B}.
Unfortunately, no good bound for the approximation constant was previously known, and indeed achieving a small constant seemed to be non-trivial.
%
In this paper, we make a substantial progress in this direction
by presenting new 4-approximation algorithms for \textsf{2EC-B} and \textsf{2EC-B-C} that run in linear time (the algorithm for \textsf{2EC-B-C} runs in linear time once the $2$-edge-connected components of $G$ are available; if not, they can be computed in $O(n^2)$ time~\cite{2CC:HenzingerKL14}).

From the practical viewpoint, we provide efficient implementations of our algorithms that are very fast in practice.
%
%
We further propose and implement several heuristics that improve the size (i.e., the number of edges) of the computed spanning subgraphs in practice.
Some of our algorithms require $O(mn)$ time in the worst case, so we also present several techniques to achieve significant speedups in their running times.
With all these implementations, we conduct a thorough experimental study and report its main findings. We believe that this is crucial to assess the merits of all the algorithms considered in practical scenarios.

The remainder of this paper is organized as follows.
We introduce some preliminary definitions and graph-theoretical terminology in Section~\ref{sec:dominators}.
Then, in Section~\ref{sec:algorithms} we 
describe our basic approaches and provide a $3/2$-approximation algorithm for \textsf{2EC-C} and $4$-approximation algorithms for \textsf{2EC-B} and \textsf{2EC-B-C}.
Our empirical study is presented in Section~\ref{sec:experimental}.
Finally, in Section~\ref{sec:concluding} we discuss some open problems and directions for future work.

\section{Preliminaries}
\label{sec:dominators}


In this section, we introduce some basic terminology that will be useful throughout the paper.

\paragraph{Flow graphs, dominators, and independent spanning trees.}
A \emph{flow graph} is a digraph such that every vertex
is reachable from a distinguished start vertex. Let $G=(V,E)$ be a
strongly connected digraph. For any vertex $s \in V$, we denote by
$G(s)=(V,E,s)$ the corresponding flow graph with start vertex $s$;
all vertices in $V$ are reachable from $s$ since $G$ is strongly
connected. The \emph{dominator relation} in $G(s)$ is defined as
follows: A vertex $u$ is a \emph{dominator} of a vertex $w$ ($u$
\emph{dominates} $w$) if every path from $s$ to $w$ contains $u$;
$u$ is a \emph{proper dominator} of $w$ if $u$ dominates $w$ and
$u \not= w$.  The dominator relation is reflexive and transitive.
Its transitive reduction is a rooted tree, the \emph{dominator
tree} $D(s)$: $u$ dominates $w$ if and only if $u$ is an ancestor
of $w$ in $D(s)$. If $w \not= s$, $d(w)$, the parent of $w$ in
$D(s)$, is the \emph{immediate dominator} of $w$: it is the unique
proper dominator of $w$ that is dominated by all proper dominators
of $w$. The dominator tree of a flow graph can be computed in
linear time, see, e.g.,~\cite{domin:ahlt,dominators:bgkrtw}. An
edge $(u,w)$ is a \emph{bridge} in $G(s)$ if all paths from $s$ to
$w$ include $(u,w)$.\footnote{Throughout,
we use consistently the term  \emph{bridge} to refer to a bridge of a flow graph $G(s)$ and the term \emph{strong bridge} to refer to a strong bridge in the original digraph $G$.}
Italiano et al. \cite{Italiano2012} showed that
%
the strong bridges of $G$
can be computed from the bridges of the flow graphs $G(s)$ and $G^R(s)$, where $s$ is an arbitrary start vertex and $G^R$ is the digraph that results from $G$ after reversing edge directions.

A spanning tree $T$ of a flow graph $G(s)$ is a tree with root $s$ that contains a path from $s$ to $v$ for all vertices $v$.
Two spanning trees $B$ and $R$ rooted at $s$ are \emph{edge-disjoint} if they have no edge in common. A flow graph $G(s)$ has two such spanning trees if and only if it has no bridges~\cite{st:t}. The two spanning trees are \emph{maximally edge-disjoint} if the only edges they have in common are the bridges of $G(s)$. Two (maximally) edge-disjoint spanning trees can be computed in linear-time by an algorithm of Tarjan~\cite{st:t}, using the disjoint set union data structure of Gabow and Tarjan~\cite{dsu:gt}.
Two spanning trees $B$ and $R$ rooted at $s$ are \emph{independent} if for all vertices $v$, the paths from $s$ to $v$ in $B$ and $R$ share only the dominators of $v$. Every flow graph $G(s)$ has two such spanning trees, computable in linear time \cite{domv:gt05,domcert} which are maximally edge-disjoint.


\begin{figure}[t]
\begin{center}
\includegraphics[trim={0 0 0 10cm}, clip=true, width=\textwidth]{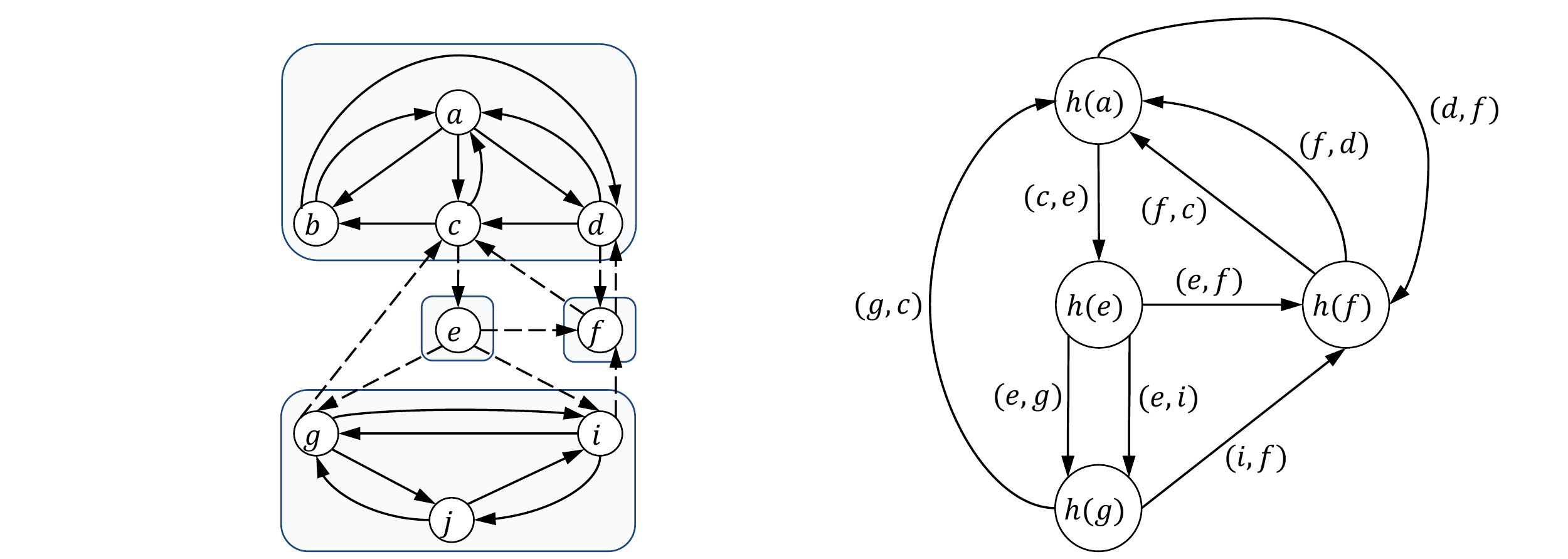}
\end{center}
\vspace*{-0.1in}
\caption{The right part shows the condensed graph $H$ of the digraph of Figure \ref{fig:2ebc}. Each edge of $H$ is labeled with the corresponding original edge of $G$.}
\label{fig:condensed}
\end{figure}

\paragraph{Condensed graph.} The \emph{condensed graph} is the digraph $H$ obtained from $G$ by contracting each $2$-edge-connected component of $G$ into a single supervertex.
Note that $H$ is a multigraph since the contractions can create loops and parallel edges; see Figure \ref{fig:condensed}.
For any vertex $v$ of $G$, we denote by $h(v)$ the supervertex of $H$ that contains $v$.
Every edge $(h(u),h(v))$ of $H$ is associated with the corresponding original edge $(u,v)$ of $G$.
Given a condensed graph $H$, we can obtain the \emph{expanded graph} by reversing the contractions; each supervertex $h(v)$ is replaced by the subgraph induced by the original vertices $u$ with $h(u)=h(v)$, and each edge $(h(u),h(v))$ of $H$ is replaced with the corresponding original edge $(u,v)$.

\section{Approximation algorithms and heuristics}
\label{sec:algorithms}

We start by describing our main approaches for solving problem
\textsf{2EC-B}. Let $G=(V,E)$ be the input directed graph.
The first two algorithms
process one edge $(x,y)$ of the current subgraph $G'$ of $G$ at a time, and test if it is safe to remove $(x,y)$. Initially $G'=G$, and the order in which the edges are processed is arbitrary. The third algorithm 
starts with the empty graph $G'=(V,\emptyset)$, and adds the edges of spanning trees of certain subgraphs of $G$ until the resulting digraph is strongly connected and has the same $2$-edge-connected blocks as $G$.

\begin{description}
  \item[Two Edge-Disjoint Paths Test.]
    We test if $G' \setminus (x,y)$ contains two edge-disjoint paths from $x$ to $y$.
    If this is the case, then we remove edge $(x,y)$.
    This test takes $O(m)$ time per edge, so the total running time is $O(m^2)$.
    We refer to this algorithm as \textsf{Test2EDP-B}.
    Note that \textsf{Test2EDP-B} computes a minimal $2$-approximate solution for the \textsf{2ECSS} problem~\cite{CT00}, which is not necessarily minimal for the \textsf{2EC-B} problem.

  \item[$2$-Edge-Connected Blocks Test.]
    If $(x,y)$ is not a strong bridge in $G'$, we test if $G' \setminus (x,y)$ has the same $2$-edge-connected blocks as $G'$. If this is the case then we remove edge $(x,y)$.
    We refer to this algorithm as \textsf{Test2ECB-B}.
    Since the $2$-edge-connected blocks of a graph can be computed in linear time~\cite{2ECB}, \textsf{Test2ECB-B} runs in $O(m^2)$ time.
    \textsf{Test2ECB-B} computes a minimal solution for \textsf{2EC-B} and achieves an approximation ratio of $4$ (see Section~\ref{section:IST-modified}).

  \item[Independent Spanning Trees.]
    We can compute a sparse certificate for $2$-edge-connected blocks as in \cite{2ECB}, based on a linear-time construction of two independent spanning trees of a flow graph \cite{domv:gt05,domcert}.
    We refer to this algorithm as \textsf{IST-B original}.
    We will show later that a suitably modified construction, which we refer to as \textsf{IST-B}, yields a linear-time $4$-approximation algorithm.
\end{description}
%
The first two approaches \textsf{Test2EDP-B} and \textsf{Test2ECB-B} can be combined into a hybrid algorithm (\textsf{Hybrid-B}),  as follows:
\begin{itemize}
\item if the tested edge $(x,y)$ connects vertices in the same $2$-edge-connected block
(i.e., $x\leftrightarrow_{\mathrm{2e}} y$), then apply \textsf{Test2EDP-B}; otherwise, apply  \textsf{Test2ECB-B}.
\end{itemize}
One can show that \textsf{Hybrid-B} returns the same sparse subgraph as \textsf{Test2ECB-B}.

\begin{lemma}
\label{lemma:Test2EDP-Test2ECB}
Let $(x,y)$ be an edge of $G$. Algorithm \textsf{Test2EDP-B} deletes $(x,y)$ only if \textsf{Test2ECB-B} does as well.
Moreover, if  $x$ and $y$ belong to the same $2$-edge-connected block of $G$, then algorithms \textsf{Test2EDP-B} and \textsf{Test2ECB-B} are equivalent for $(x,y)$,
i.e., edge $(x,y)$ is deleted by \textsf{Test2ECB-B} if and only if it is deleted by \textsf{Test2EDP-B}.
\end{lemma}
\begin{proof}
To prove the first part of the lemma, suppose that $(x,y)$ is deleted by \textsf{Test2EDP-B}. We show that the $2$-edge-connected blocks of $G$ are not affected by this deletion.
Consider any pair of $2$-edge-connected vertices $u$ and $w$ that was affected by the deletion of $(x,y)$, that is, the number of edge-disjoint paths from $u$ to $w$ was reduced.
Let $(U,W)$ be a minimum $u$-$w$ cut in $G\setminus (x,y)$, i.e., $U, W \subseteq V$, $U \cap W = \emptyset$, $u \in U$ and $w \in W$.
Then we also have $x \in U$ and $y \in W$. Since $G\setminus (x,y)$  has at least two edge-disjoint paths from $x$ to $y$, %
Menger's theorem implies that there are at least two edges directed from $U$ to $W$.
Thus, Menger's theorem implies $G\setminus (x,y)$  has at least two edge-disjoint paths from $u$ to $w$.

We now prove the second part of the lemma. Suppose that $x$ and $y$ lie in the same $2$-edge-connected block of $G$, and edge $(x,y)$ is deleted by algorithm \textsf{Test2ECB-B}.
This implies that $G\setminus (x,y)$ has two edge-disjoint paths from $x$ to $y$, so algorithm \textsf{Test2EDP-B} would also delete $(x,y)$.
\end{proof}

\subsection{Providing a sparse certificate as input}
\label{sec:sparse-input}

As we mentioned above, algorithm \textsf{IST-B} computes in linear time a sparse certificate for the $2$-edge-connected blocks of an input digraph $G$, i.e., a spanning subgraph of $G$ with $O(n)$ edges that has the same $2$-edge-connected blocks with $G$. In order to speed up our slower heuristics, \textsf{Test2EDP-B}, \textsf{Test2ECB-B} and \textsf{Hybrid-B}, we can apply them on the sparse certificate instead of the original digraph, thus reducing their running time from $O(m^2)$ to $O(n^2)$. Moreover, given that \textsf{IST-B} achieves a $4$-approximation (Theorem~\ref{theorem:ApproximationRatio}), it follows that \textsf{Test2EDP-B}, \textsf{Test2ECB-B} and \textsf{Hybrid-B}
produce a $4$-approximation for \textsf{2EC-B} in $O(n^2)$ time. Therefore, we applied this idea in all our implementations. See Table \ref{tab:algorithms} in Section \ref{sec:experimental}.
We also note that for the tested inputs, the quality of the computed solutions was not affected significantly
by the fact that we applied the heuristics on the sparse certificate computed by \textsf{IST-B} instead of the original digraph. Indeed, on average, the number of edges in the computed subgraph
was reduced by $6\%$ for \textsf{Test2EDP-B} and increased by less than $0.5 \%$  for \textsf{Hybrid-B}. The speed up gained, on the other hand, was by a factor slightly less than $5$ for \textsf{Test2EDP-B} and by a factor slightly larger than $2$ for \textsf{Hybrid-B}.

\ignore{
Note that \textsf{Test2EDP-B}, \textsf{Test2ECB-B} and \textsf{Hybrid-B} produce the same approximation ratio with \textsf{IST-B} for \textsf{2EC-B} in $O(n^2)$ time
if they are run on the sparse subgraph computed by \textsf{IST-B} instead of the original digraph.
Given that \textsf{IST-B} achieves a 4-approximation (Theorem~\ref{theorem:ApproximationRatio}), it follows that \textsf{Test2EDP-B}, \textsf{Test2ECB-B} and \textsf{Hybrid-B}
produce a $4$-approximation for \textsf{2EC-B} in $O(n^2)$ time.
Therefore, we applied this idea in all our implementations.
See Table \ref{tab:algorithms} in Section \ref{sec:experimental}. We note that for the tested inputs, the quality of the computed solutions was not affected significantly
by the fact that we applied the heuristics on the subgraph computed by \textsf{IST-B} instead of the original graph. Indeed, on average, the number of edges in the computed subgraph
was reduced by $6\%$ for \textsf{Test2EDP-B} and increased by less than $5 \tcperthousand$  for \textsf{Hybrid-B}, while the speedup attained was by a factor close to $5$ for \textsf{Test2EDP-B} and by a factor slightly larger than $2$ for \textsf{Hybrid-B}. \sidenote{LOUKAS: Added these averages.}
}

\subsection{Maintaining the $2$-edge-connected components}
\label{sec:2-edge-components}

Although all the above algorithms do not maintain the $2$-edge-connected components of the original graph, we can still apply them to get an approximation for \textsf{2EC-B-C}, as follows.
First, we compute the $2$-edge-connected components of $G$
and solve the \textsf{2ECSS} problem independently for each such component. Then, we can apply any of the algorithms for \textsf{2EC-B} (\textsf{Test2EDP-B}, \textsf{Test2ECB-B}, \textsf{Hybrid-B} or \textsf{IST-B}) for the edges that connect different components.
To speed them up,  we apply them to the condensed graph $H$ of $G$.  Let $H'$ be the subgraph of $H$ computed by any of the above heuristics, and let $G'$ be the expanded graph of $H'$, were we replace each supervertex of $H$ with the corresponding $2$-edge-connected sparse subgraph computed before.
We refer to the corresponding algorithms obtained this way as \textsf{Test2EDP-BC}, \textsf{Test2ECB-BC}, \textsf{Hybrid-BC} and \textsf{IST-BC}.
The next lemma shows that indeed $G'$ is a valid solution to the \textsf{2EC-B-C} problem.

\begin{lemma}
\label{lemma:condensed}
Digraph $G'$ is strongly connected and has the same $2$-edge-connected components and blocks as $G$.
\end{lemma}
\begin{proof}
Digraph $G'$ is strongly connected because the algorithms do not remove strong bridges. It is also clear that $G'$ and $G$ have the same $2$-edge-connected components. So it remains to consider the $2$-edge-connected blocks. Let $u$ and $w$ be two arbitrary vertices of $G$. We show that $u$ and $w$ are $2$-edge-connected in $G'$ if and only if they are $2$-edge-connected in $G$.
The ``only if'' direction follows from the fact that $G'$ is a subgraph of $G$. We now prove the ``if'' direction. Suppose $u$ and $w$ are $2$-edge-connected in $G$.
If $u$ and $w$ are located in the same $2$-edge-connected component then obviously they are $2$-edge-connected in $G'$. Suppose now that $u$ and $w$ are located in different components, so $h(u) \not= h(w)$. By construction,  for any $(S,T)$ cut in $H$ such that $h(u) \in S$ and $h(v) \in T$ there are at least two edges directed from $S$ to $T$ and at least two edges directed from $T$ to $S$. This property is maintained by all algorithms, so it also holds in $H'$.
Then, for any $(U,W)$ cut in the expanded graph $G'$ such that $u \in U$ and $w \in W$ there are at least two edges directed from $U$ to $W$ and at least two edges directed from $W$ to $U$. So $u$ and $w$ are $2$-edge-connected in $G'$ by Menger's theorem.
\end{proof}

As a special case of applying \textsf{Test2EDP-B} to $H$, we can immediately remove loops and parallel edges $(h(u),h(v))$ if $H$ has more than two edges directed from $h(u)$ to $h(v)$.
To obtain faster implementations,
we solve the \textsf{2ECSS} problems in linear-time using edge-disjoint spanning trees~\cite{edge-disjoint:edmonds,st:t}.
Let $C$ be a $2$-edge-connected component of $G$. 
We select an arbitrary vertex $v \in C$ as a root and compute two edge-disjoint spanning trees in the flow graph $C(v)$
and  two edge-disjoint spanning trees in the reverse flow graph $C^R(v)$. The edges of these spanning trees give a $2$-approximate solution $C'$ for \textsf{2ECSS} on $C$.
Moreover, as in \textsf{2EC-B}, we can apply algorithms \textsf{Test2EDP-BC}, \textsf{Test2ECB-BC} and \textsf{Hybrid-BC} on the sparse subgraph computed by \textsf{IST-BC}.
Then, these algorithms produce a $4$-approximation for \textsf{2EC-B-C} in $O(n^2)$ time. 
Furthermore, for these $O(n^2)$-time algorithms, we can improve the approximate solution $C'$ for \textsf{2ECSS} on each $2$-edge-connected component $C$ of $G$, by applying
the two edge-disjoint paths test on the edges of $C'$.
We incorporate all these ideas in all our implementations.

We can also use the condensed graph in order to obtain an efficient approximation algorithm for \textsf{2EC-C}.
To that end, we can apply the algorithm of Laehanukit et al.~\cite{LGS:MSCS:2012} and get a $3/2$-approximation of the 
\textsf{2ECSS} problem independently for each $2$-edge-connected component of $G$.
Then, since we only need to preserve the strong connectivity of $H$, we can run the algorithm of Vetta~\cite{Vetta:MSCS:2001} on a digraph $\tilde{H}$ that results from $H$ after removing all loops and parallel edges.
This computes a spanning subgraph $H'$ of $\tilde{H}$ that is a $3/2$-approximation for \textsf{SCSS} in $H$.
The corresponding expanded graph $G'$, where we substitute each supervertex $h(v)$ of $H$ with the approximate smallest \textsf{2ECSS}, gives a $3/2$-approximation for \textsf{2EC-C}.
A faster and deterministic $2$-approximation algorithm for \textsf{2EC-C} can be obtained as follows.
For the \textsf{2ECSS} problems we use the edge-disjoint spanning trees $2$-approximation algorithm described above.
Then, we solve \textsf{SCSS} on $\tilde{H}$ by applying the linear-time algorithm of Zhao et al.~\cite{ZNI:MSCS:2003}.
This yields a $2$-approximation algorithm for \textsf{2EC-C} that runs in linear time once the $2$-edge-connected components of $G$ are available (if not, they can be computed in $O(n^2)$ time~\cite{2CC:HenzingerKL14}).
We refer to this algorithm as \textsf{ZNI-C}.

\begin{theorem}
\label{theorem:2EC-C}
There is a polynomial-time algorithm for \textsf{2EC-C} that achieves an approximation ratio of $3/2$. Moreover, if the $2$-edge-connected components of $G$ are available, then we can compute a $2$-approximate \textsf{2EC-C} in linear time.
\end{theorem}

\subsection{Independent Spanning Trees}
\label{section:IST}

Here we present our new algorithm \textsf{IST-B} and prove that it
gives a linear-time $4$-approximation for \textsf{2EC-B} and
\textsf{2EC-B-C}. Since \textsf{IST-B} is a modified version of
the sparse certificate $C(G)$ for the $2$-edge-connected blocks of
a digraph $G$ \cite{2ECB} (\textsf{IST-B original}), let us review
\textsf{IST-B original} first.

Let $s$ be an arbitrarily chosen start vertex of the strongly connected digraph $G$.
The \emph{canonical decomposition} of the dominator tree $D(s)$ is the forest of rooted trees that results from $D(s)$ after the deletion of all the bridges of $G(s)$.
Let $T(v)$ denote the tree containing vertex $v$ in this decomposition. We refer to the subtree roots in the canonical decomposition as \emph{marked vertices}.
For each marked vertex $r$ we define the \emph{auxiliary graph $G_r = (V_r, E_r)$ of $r$} as follows.
\begin{itemize}
\item The vertex set $V_r$ of $G_r$ consists of all the vertices in $T(r)$, referred to as \emph{ordinary} vertices, and a set of
\emph{auxiliary} vertices, which are obtained by contracting vertices in $V\setminus T(r)$, as follows.
\begin{itemize}
\item Let $v$ be a vertex in $T(r)$.
We say that $v$ is a \emph{boundary vertex in} $T(r)$ if $v$ has a marked child in $D(s)$.
Let $w$ be a marked child of a boundary vertex $v$:
all the vertices that are descendants of  $w$ in $D(s)$  are contracted into $w$.
\end{itemize}
\item All vertices in $V \setminus T(r)$ that are not descendants of $r$ are contracted into $d(r)$ ($r \not= s$ if any such vertex exists).
\end{itemize}
Figures~\ref{fig:l1-auxiliary} and \ref{fig:l2-auxiliary} illustrate the canonical decomposition of a dominator tree and the corresponding auxiliary graphs. 

\begin{figure}[t!]
\begin{center}
\includegraphics[trim={0 0 0 7cm}, clip=true, width=1.0\textwidth]{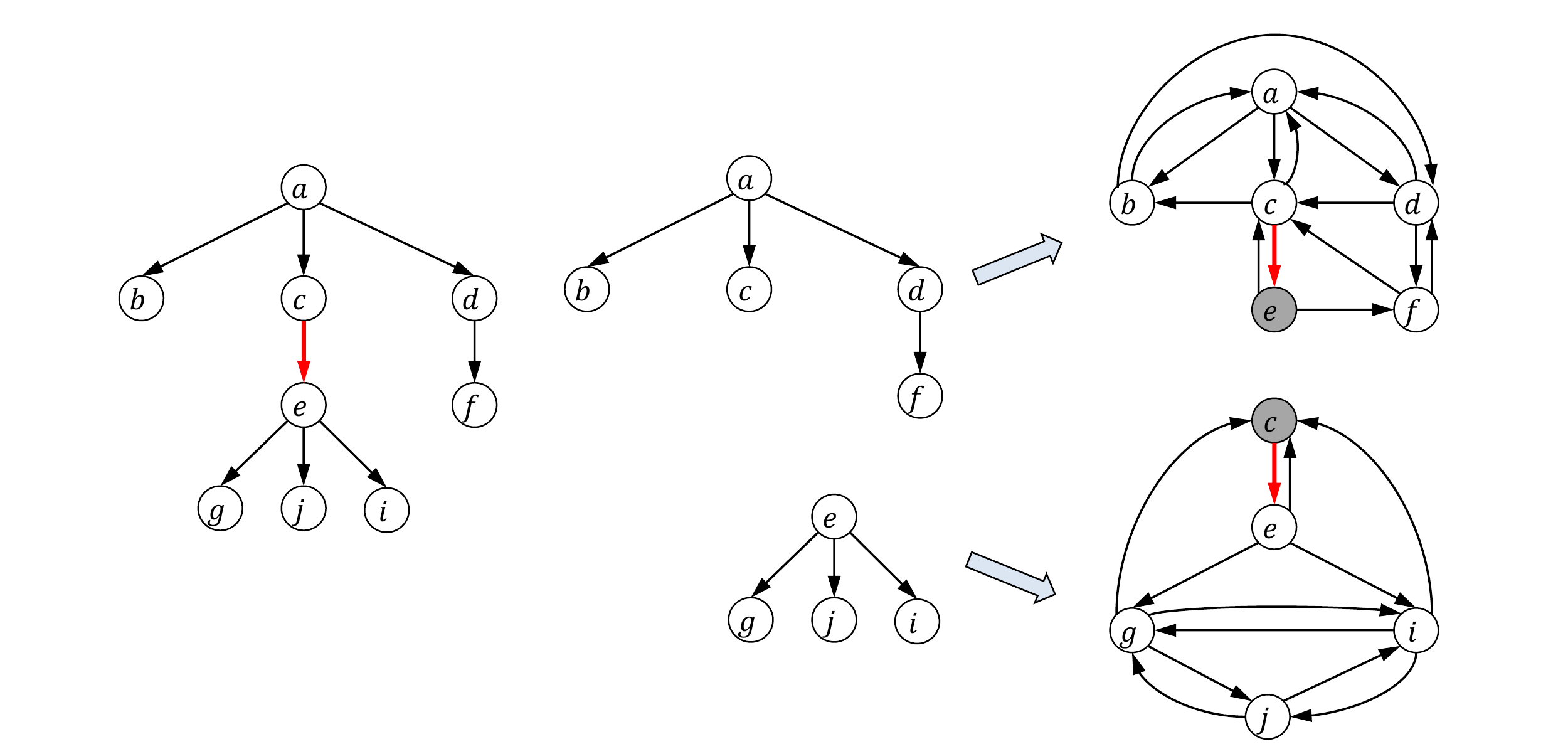}
\end{center}
\vspace{-0.2cm}
\hspace{2.8cm} (a) \hspace{3.9cm} (b) \hspace{4.8cm} (c)
\caption{(a) The dominator tree $D(a)$ of the flow graph of Figure \ref{fig:2ebc} with start vertex $a$. The strong bridge $(c,e)$, shown in red (better viewed in color), appears as an edge of the dominator tree. (b) The subtrees $T(a)$ and $T(e)$ of the canonical decomposition of $D(a)$ after the deletion of $(c,e)$, and (c) their corresponding first-level auxiliary graphs $G_a$ and $G_e$. Auxiliary vertices are shown grey.
}
\label{fig:l1-auxiliary}
\end{figure}
\begin{figure}[h!]
\begin{center}
\includegraphics[trim={0 0 0 10cm}, clip=true, width=1.0\textwidth]{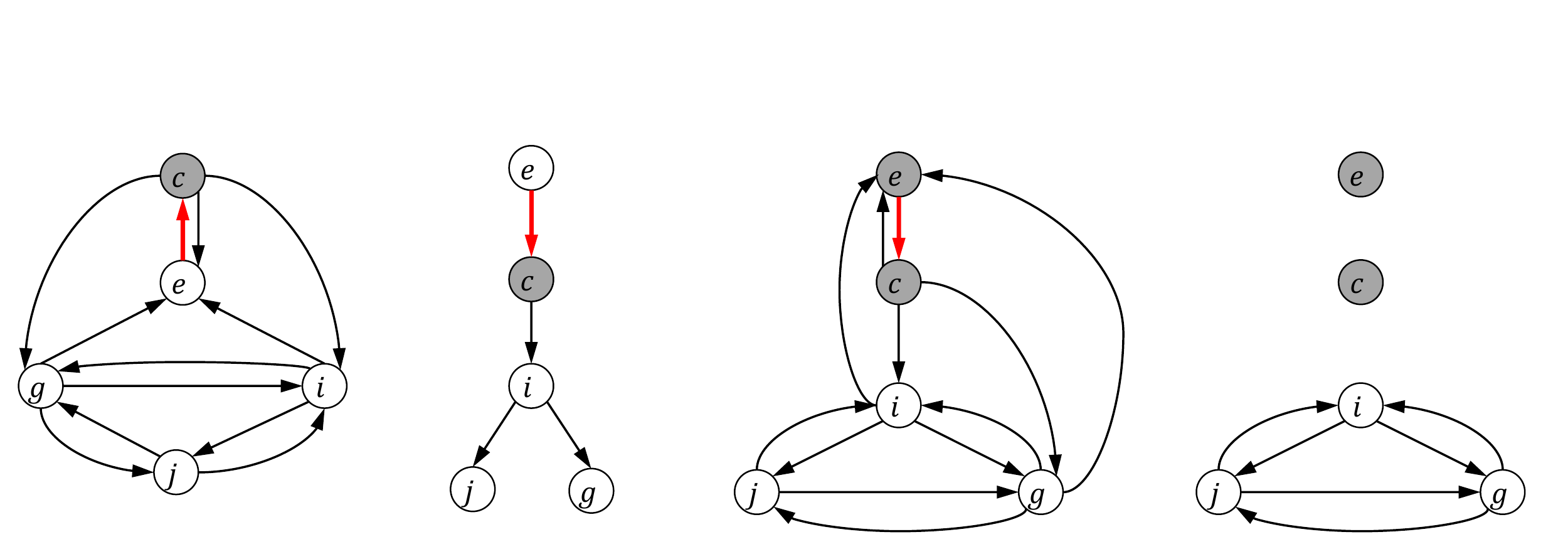}
\end{center}
\vspace{-0.5cm}
\hspace{1.5cm} (a) \hspace{3cm} (b) \hspace{3.1cm} (c) \hspace{4.1cm} (d)
\caption{(a) The reverse graph $H^R$ of the auxiliary graph $H=G_e$ of Figure \ref{fig:l1-auxiliary}. The strong bridge $(c,e)$ of the original digraph, shown in red (better viewed in color), appears as the strong bridge $(e,c)$ in $H^R$. (b) The dominator tree of $H^R(e)$ with start vertex $e$. (c) The second-level auxiliary graph $H^R_c$. Auxiliary vertices are shown grey. (d) The strongly connected components of $H^R_c \setminus (e,c)$.
The strongly connected component $\{i,j,g\}$ is a $2$-edge-connected block of the original digraph.
}
\label{fig:l2-auxiliary}
\end{figure}

During those contractions, parallel edges are eliminated. We call an edge in $E_r \setminus E$ \emph{shortcut edge}. Such an edge has an auxiliary vertex as an endpoint.
We associate each shortcut edge $(u,v) \in E_r$  with a corresponding original edge $(x,y) \in E$, i.e., $x$ was contracted into $u$ or $y$ was contracted into $v$ (or both).
If $G(s)$ has $b$ bridges then all the auxiliary graphs $G_r$ have at most $n+2b$ vertices and $m+2b$ edges in total and can be computed in $O(m)$ time.
As shown in \cite{2ECB}, two ordinary vertices of an auxiliary graph $G_r$ are $2$-edge-connected in $G$ if and only if they are $2$-edge-connected in $G_r$.
Thus the $2$-edge-connected blocks of $G$ are a refinement of the vertex sets in the trees of the canonical decomposition.
The sparse certificate of \cite{2ECB} is constructed in three phases.
We maintain a list (multiset) $L$ of the edges to be added in $C(G)$; initially $L=\emptyset$. The same edge may be inserted into $L$ multiple times, but the total number of insertions will be $O(n)$. So
the edges of $C(G)$ can be obtained from $L$ after we remove duplicates, e.g. by  using radix sort. Also, during the construction, the algorithm may choose a shortcut edge or a reverse edge to be inserted into $L$. In this case we insert the associated original edge instead.

\begin{description}
\item[Phase 1.] We insert into $L$ the edges of two independent
spanning trees, $B(G(s))$ and $R(G(s))$ of $G(s)$.

\item[Phase 2.] For each auxiliary graph $H=G_r$ of $G(s)$, that we
refer to as the \emph{first-level auxiliary graphs}, we compute
two independent spanning trees $B(H^R(r))$ and $R(H^R(r))$ for the
corresponding reverse flow graph $H^R(r)$ with start vertex $r$.
We insert into $L$ the edges of these two spanning trees. We note
that $L$ induces a strongly connected spanning subgraph of $G$ at
the end of this phase.

\item[Phase 3.] Finally, in the third phase we 
process the \emph{second-level auxiliary graphs}, which are the auxiliary graphs of $H^R$ for all first-level auxiliary graphs $H$.
Let $(p,q)$ be a bridge of $H^R(r)$, and let $H_q^R$ be the corresponding second-level auxiliary graph.
For every strongly connected component $S$ of $H_q^R \setminus (p,q)$, we choose an arbitrary vertex $v \in S$ and compute a spanning tree of $S(v)$ and a spanning tree of $S^R(v)$, and insert their edges into $L$; see Figure~\ref{fig:l2-auxiliary}.
\end{description}

The above construction inserts $O(n)$ edges into $C(G)$, and therefore achieves a constant approximation ratio for \textsf{2EC-B}. It is not straightforward, however, to give a good bound for this constant, since the spanning trees that are used in this construction contain auxiliary vertices that are created by applying two levels of the canonical decomposition. In the next section we analyze a modified version of the sparse certificate construction, and show that it achieves a $4$-approximation for \textsf{2EC-B}. Then we show that we also achieve a $4$-approximation for \textsf{2EC-B-C} by applying this sparse certificate on the condensed graph $H$.

\subsubsection{The new algorithm \textsf{IST-B}}
\label{section:IST-modified}

The main idea behind \textsf{IST-B} is to limit the number of edges added to the sparse certificate $C(G)$ because of auxiliary vertices.
In particular, we show that in Phase 2 of the construction it suffices to add at most one new edge for each first-level auxiliary vertex,
while in Phase 3 at most $2b$ additional edges are necessary for all second-level auxiliary vertices, where $b$ is the number of bridges in $G(s)$.

We will use the following lemma about the strong bridges in auxiliary graphs,
which implies that for any second-level auxiliary vertex $x$ that was not an auxiliary vertex in the first level,
subgraph $C(G)$ contains the unique edge leaving $x$ in $H$.

\begin{lemma}
\label{lemma:auxiliary-graph-strong-bridges}
Let $(u,v)$ be a strong bridge of a first-level auxiliary graph $H=G_r$ that is not a bridge in $G(s)$. Then $(v,u)$ is a bridge in the flow graph $H^R(r)$.
\end{lemma}
\begin{proof}
Consider the dominator tree $D_H(r)$ of the flow graph $H(r)$. Let $D'$ be the tree that results from $D_H(r)$ after the deletion of the auxiliary vertices.
Then we have $D'=T(r)$.
Moreover, for each auxiliary vertex $x \not= d(r)$, $(d(x),x)$ is the unique edge entering $x$ in $H$, which is a bridge in $G(s)$.
Also, $(d(r),r)$ is the unique edge leaving $d(r)$ in $H$ which too is a bridge in $G(s)$.
By \cite{Italiano2012} we have that a strong bridge of $H$ must appear as a bridge of $H(r)$ or as the reverse of a bridge in $H^R(r)$, so the lemma follows.
\end{proof}

First we will describe our modified construction and apply a charging scheme for the edges added to $C(G)$ that are adjacent to auxiliary vertices.
Then, we use this scheme to prove that the modified algorithm achieves the desired $4$-approximation.
Phase~1 remains the same and we explain the necessary modifications for Phases~2 and 3.

\begin{description}
\item[Modified Phase 2.] Let $H=G_r$ be a first-level auxiliary graph.
In the sparse certificate we include two independent spanning trees, $B(H^R(r))$ and $R(H^R(r))$, of the reverse flow graph $H^R(r)$ with start vertex $r$.
In our new construction, each auxiliary vertex $x$ in $H^R$ will contribute at most one new edge in $C(G)$.
Suppose first that
$x=d(r)$, which exists if $r \not=s$. The only edge entering $d(r)$ in $H^R$ is $(r,d(r))$ which is the reverse edge of the bridge $(d(r),r)$ of $G(s)$. So $d(r)$ does not add a new edge in $C(G)$, since all the bridges of $G(s)$ were added in the first phase of the construction.
Next we consider an auxiliary vertex $x \not= d(r)$. In $H^R$ there is a unique edge $(x,z)$ leaving $x$, where $z=d(x)$.
This edge is the  reverse of the bridge $(d(x),x)$ of $G(s)$.
Suppose that $x$ has no children in $B(H^R(r))$ and $R(H^R(r))$.
Deleting $x$ and its two entering edges in both spanning trees does not affect the existence of two edge-disjoint paths from $v$ to $r$ in $H$, for any ordinary vertex $v$.
However, the resulting graph $C(G)$ at the end may not be strongly connected.
To fix this, it suffices to include in $C(G)$ the reverse of an edge entering $x$ from only one spanning tree.
Finally, suppose that $x$ has children, say in $B(H^R(r))$.
Then $z=d(x)$ is the unique child of $x$ in $B(H^R(r))$, and the reverse of the edge $(x,z)$ of $B(H^R(r))$ is already included in $C(G)$ by Phase 1.
Therefore, in all cases, we can charge to $x$ at most one new edge.

\item[Modified Phase 3.]
Let $H_q^R$ be a second-level auxiliary graph of $H^R$.
Let $e$ be the strong bridge entering $q$ in $H^R$, and let $S$ be a strongly connected component in $H_q^R \setminus e$.
In our sparse certificate we include the edges of a strongly connected subgraph of $S$, 
so we have spanning trees $T$ and $T^R$ of 
$S(v)$ and $S^R(v)$, respectively, rooted at an arbitrary ordinary vertex $v$.
Let $x$ be an auxiliary vertex of $S$. We distinguish two cases:
\begin{itemize}
\item[(i)] If $x$ is a first-level auxiliary vertex in $H$ then it has a unique entering edge $(w,x)$ which is a bridge in $G(s)$ already included in $C(G)$.
\item[(ii)] If $x$ is ordinary in $H$ but a second-level auxiliary vertex in $H_q$ then it has a unique leaving edge $(x,z)$, which, by Lemma~\ref{lemma:auxiliary-graph-strong-bridges}, is a bridge in $H^R(r)$
and $C(G)$ already contains a corresponding original edge.
\end{itemize}
Consider the first case. If $x$ is a leaf in $T^R$ then we can delete the edge entering $x$ in $T^R$.
Otherwise, $w$ is the unique child of $x$ in $T^R$, and the corresponding edge $(w,x)$ entering $x$ in $H$ has already been inserted in $C(G)$.
The symmetric arguments hold if $x$ is ordinary in $H$.

This analysis implies that we can associate each second-level auxiliary vertex with one edge in each of 
$T$ and $T^R$ that is either not needed in $C(G)$ or has already been inserted.
If all such auxiliary vertices are associated with distinct edges then they do not contribute any new edges in $C(G)$.
Suppose now that there are two second-level auxiliary vertices $x$ and $y$ that are associated with a common edge $e$.
This can happen only if one of these vertices, say $y$, is a first-level auxiliary vertex, and $x$ is ordinary in $H$.
Then $y$ has a unique entering edge in $H$, which means that $e=(x,y)$ is a strong bridge, and thus already in $C(G)$.
Also $e \in T$ and $e^{R}=(y,x) \in T^R$.
In this case, we can treat $x$ and $y$ as a single auxiliary vertex that results from the contraction of $e$, which contributes at most two new edges in $C(G)$.
Since $y$ is a first-level auxiliary vertex, this can happen at most $b$ times in all second-level auxiliary graphs, so a bound of $2b$ such edges follows.
\end{description}

Using the above construction we can now prove that our modified version of the sparse certificate achieves an approximation ratio of $4$.

\begin{theorem}
\label{theorem:ApproximationRatio}
There is a linear-time approximation algorithm for the \textsf{2EC-B} problem that achieves an approximation ratio of $4$.
Moreover, if the $2$-edge-connected components of the input digraph are known in advance, we can compute a $4$-approximation for the \textsf{2EC-B-C} problem in linear time.
\end{theorem}
\begin{proof}
Let $b$ denote (as above)  the number of bridges in the flow graph $G(s)$. Note that $b \le n-1$.
We consider the three phases of the construction of $C(G)$ separately and account for the new edges that are added in each phase.
Consider the two independent spanning trees $B$ and $R$ of $G(s)$ that are computed in the first phase. If an edge $(u,v)$ is a bridge in $G(s)$ then it is the unique edge entering $v$ in $B \cup R$. Thus these two independent spanning trees add into $L$ exactly $2(n-b-1)+b=2n-b-2$ edges.

Now we consider the Modified Phase 2. Let $H=G_r$ be a first-level auxiliary graph. Let $o_r$ and $a_r$ be, respectively, the number of ordinary and auxiliary vertices in $G_r$.
In the sparse certificate we include two independent spanning trees, $B(H^R(r))$ and $R(H^R(r))$, of the reverse flow graph $H^R(r)$ with start vertex $r$.
As already explained in the analysis of this phase, each auxiliary vertex $x$ in $H^R$ may contribute at most one new edge in $C(G)$.
Since $r$ and $d(r)$ do not contribute any new edges, the total number of edges added for $H$ is at most $2(o_r-1)+(a_r-1)$. Hence, the total number of edges added during the second phase is at most $\sum_{r}(2o_r+a_r-3)$, where the sum is taken over all $b+1$ marked vertices $r$. Observe that $\sum_r (o_r) = n$ and $\sum_r (a_r) = 2b$, so we have $\sum_{r}(2o_r+a_r-3) \le 2n + 2b -3b = 2n -b$.
We note that, as in the original construction, $C(G)$ is strongly connected at the end of this phase. Moreover, in this phase we include in $L$ the strong bridges of $G$ that are not bridges in $G(s)$.

It remains to account for the edges added during the third phase.
Here we consider the strongly connected components for each auxiliary graph $H_q^R$ of $H^R$ after removing the strong bridge entering $q$ in $H^R$.
By the argument in the description of the Modified Phase 3, the second-level auxiliary vertices contribute at most $2b$ new edges in total.

We note that the $2$-edge-connected blocks of $G$ are formed by the ordinary vertices in each strongly connected component computed for the second-level auxiliary graphs.
Consider such a strongly connected component $S$. Let $o_S$ be the number of ordinary vertices in $S$. If $o_S \le 1$ then we do not include any edges for $S$. So suppose that $o_S \ge 2$.
Excluding the at most $2b$ additional edges, the auxiliary vertices in $S$ do not contribute any new edges. So the number of edges added by $S$ is bounded by $2o_S$.
Then, the third phase adds $2n'+2b$ edges in total, where $n'=\sum_S o_S$ and the sum is taken over all strongly connected components with $o_S \ge 2$.

Overall, the number of edges added in $C(G)$ is at most $(2n - b - 2) + (2n -b) + (2n'+2b) = 4n - 2+ 2n' \le 4(n+n')$.
Next, we observe that these $n'$ vertices must have indegree and outdegree at least equal to $2$ in any solution to the \textsf{2EC-B} problem.
The remaining $n-n'$ vertices must have indegree and outdegree at least equal to one, since the spanning subgraph must be strongly connected.
Therefore, the smallest \textsf{2EC-B} has at least $(n-n') + 2n' = n + n'$ edges. The approximation ratio of $4$ follows.

Now consider the \textsf{2EC-B-C} problem, where we apply our new sparse certificate on the condensed graph $H$.
Let $k$ be the number of edges computed by our algorithm for all $2$-edge-connected components, where we apply the edge-disjoint spanning trees construction.
Let $k^{\ast}$ be the total number of edges in an optimal solution. Then $k \le 2 k^{\ast}$.
Suppose that the condensed graph has $N$ vertices.
By the previous analysis, we have that the sparse certificate of $H$ has less than $4(N+N')$ edges, where $N'$ is the number of vertices in nontrivial blocks in the condensed graph.
So, our algorithm computes a sparse certificate for $G$ with less than $4(N+N')+k \le 4(N+N') + 2k^{\ast} < 4 (N+N' + k^{\ast})$ edges.
The smallest \textsf{2EC-B-C} has at least $N+N'+k^{\ast}$, so the approximation ratio of $4$ follows.
\end{proof}

Next we note that the above proof implies that the \textsf{Test2ECB} algorithms also achieve a $4$-approximation even when they are run on the original digraphs instead of the sparse certificates.

\begin{corollary}
\label{corollary:Test2CB-B}
Algorithm \textsf{Test2ECB-B} (resp., \textsf{Test2ECB-BC}) applied on the original input (resp., condensed) graph gives a $4$-approximate solution for \textsf{2EC-B} (resp., \textsf{2EC-B-C}).
\end{corollary}
\begin{proof}
We consider first algorithm \textsf{Test2ECB-B} for the \textsf{2EC-B} problem.
Let $G$ be a strongly connected digraph with $n$ vertices, and let $n'$ be the number of vertices in nontrivial blocks (i.e., $2$-edge-connected blocks of size at least $2$).
Let $G'$ be the spanning subgraph of $G$ produced by running \textsf{Test2ECB-B} on $G$. It suffices to argue that $G'$ contains less than $4(n+n')$ edges.
Suppose that we run \textsf{IST-B} on $G'$. Let $G''$ be the resulting subgraph of $G'$. Then, $G''$ is also a solution to \textsf{2EC-B} for $G$, and by the proof of Theorem \ref{theorem:ApproximationRatio} it has at most $4(n+n')$ edges. But since $G'$ is a minimal solution to \textsf{2EC-B} for $G$, we must have $G'=G''$.

For the \textsf{2EC-B-C} problem, assume that the edge-disjoint spanning trees construction produces $k$ edges.
Then $k \le 2 k^{\ast}$, where $k^{\ast}$ is number of edges in an optimal solution.
Let $H$ be the condensed graph of $G$, and let $N$ be the number of its vertices.
Let $H'$ be the spanning subgraph of $H$ produced by running \textsf{Test2ECB-BC} on $H$. By the proof of Theorem \ref{theorem:ApproximationRatio} and the same argument as for the \textsf{2EC-B} problem, we have that $H'$  contains at most $4(N+N')$ edges, where $N'$ is the total number of vertices in nontrivial blocks of $H$.
So the corresponding expanded graph has at most $4(N+N')+k < 4(N+N'+k^{\ast})$ edges.
Since the smallest \textsf{2EC-B-C} solution has at least $N+N'+k^{\ast}$ edges, the 4-approximation follows.
\end{proof}

\subsection{Implementation details}
\label{section:IST-implementation}

Here we provide some implementation details for our algorithms.
In order to obtain a more efficient implementation of the \textsf{IST} algorithms that achieve better quality ratio in practice,
we try to reuse as many edges as possible when we build the spanning trees in the three phases of the algorithm.
In the third phase of the construction we need to solve the smallest \textsf{SCSS} problem for each subgraph $H_S$ induced by a strongly connected component $S$ in the second-level auxiliary
graphs after the deletion of a strong bridge.
To that end, we apply a modified version of the linear-time $5/3$-approximation algorithm of Zhao et al.~\cite{ZNI:MSCS:2003}.
This algorithm computes a  \textsf{SCSS} of a strongly connected digraph by performing a depth-first search (DFS) traversal.
During the DFS traversal, any cycle that is detected is contracted into a single vertex.
We modify this approach so that we can avoid inserting new edges into the sparse certificate as follows.
Since we only care about the ordinary vertices in $S$, we can construct a subgraph of $S$ that contains edges already added in $C(G)$.
We compute the strongly connected components of this subgraph and contract them. Then we apply the algorithm of Zhao et al. on the contracted graph of $S$.
Furthermore, during the DFS traversal we give priority to edges already added in $C(G)$.

We can apply a similar idea in the second phase of the construction as well.
The algorithm of \cite{domv:gt05} for computing two independent spanning trees of a flow graph uses the edges of a DFS spanning tree, together with at most $n-1$ other edges.
Hence, we can modify the DFS traversal so that we give priority to edges already added in $C(G)$.

\subsection{Heuristics applied on auxiliary graphs}
\label{section:auxiliary}

To speed up algorithms from the \textsf{Test2EDP} and
\textsf{Hybrid} families, we applied them
 to the first-level and second-level auxiliary graphs. 
Since auxiliary graphs are supposed to be smaller than the original graph, one could expect to obtain some performance gain at the price of a slightly worse approximation.
However, this performance gain cannot be taken completely for granted, as auxiliary vertices and shortcut edges 
may be repeated in several auxiliary graphs.
Our experiments indicated that applying this heuristic to second-level auxiliary graphs yields better results than the ones obtained on first-level auxiliary graphs.
We refer to those variants as
\begin{itemize}
\item \textsf{Test2EDP-B-Aux} and \textsf{Hybrid-B-Aux},
\item \textsf{Test2EDP-BC-Aux} and \textsf{Hybrid-BC-Aux},
\end{itemize}
depending on the algorithm (\textsf{Test2EDP} or \textsf{Hybrid}) and problem (\textsf{2EC-B} or \textsf{2EC-B-C}) considered.

\subsection{Trivial edges}
\label{sec:add-speed-up}

For the algorithms of the \textsf{Test2EDP} and \textsf{Hybrid} families
we use an additional speed-up heuristic in order to avoid testing edges that trivially
belong to the computed solution. We say that $(x,y)$ is a  \emph{trivial edge} of the current graph $G'$ if
it satisfies one of the following conditions:
\begin{itemize}
\item $x$ belongs to a $2$-edge-connected block of size at least two (nontrivial block) and has outdegree at most two,
    or $y$ belongs to a $2$-edge-connected block of size at least two (nontrivial block) and has indegree at most two;
\item $x$ belongs to a $2$-edge-connected block of size one (trivial block) and has outdegree one,
    or $y$ belongs to a $2$-edge-connected block of size one (trivial block) and has indegree one.
\end{itemize}
Clearly, the removal of a trivial edge will result in a digraph that either has different $2$-edge-connected blocks
or is not strongly connected.
Therefore these edges should remain in $G'$.
As we show later in our experiments, such a simple test can yield significant performance gains.

\section{Experimental analysis}
\label{sec:experimental}

We implemented the algorithms 
previously described: $7$ for \textsf{2EC-B}, $6$  for \textsf{2EC-B-C}, and one for \textsf{2EC-C}, as summarized in Table \ref{tab:algorithms}.
All implementations were written in {\tt C++} and compiled with {\tt g++ v.4.4.7} with 
flag {\tt -O3}.
We performed our experiments on a GNU/Linux machine, with Red Hat Enterprise Server v6.6: a PowerEdge T420 server 64-bit NUMA
with two Intel Xeon E5-2430 v2 processors  and 16GB of RAM RDIMM memory. Each processor has 6 cores sharing a 15MB L3 cache, and each core has a 2MB private L2 cache and 2.50GHz speed. In our experiments we did not use any parallelization, and each algorithm ran on a single core.
We report CPU times measured with the \texttt{getrusage} function. All our running times 
were averaged over ten different runs.

\begin{table}[!ht]
\begin{adjustbox}{width=\textwidth,totalheight=\textheight,keepaspectratio}
\begin{small}
\tabulinesep=1.0mm
\setlength{\tabcolsep}{4.0pt} 
\begin{tabu}[t]{p{2.78cm}p{1.32cm}p{8.2cm}p{1.50cm}} 
\hline
  Algorithm & Problem & Technique & Time \\
 \hline
 \textsf{ZNI-C}  & \textsf{2EC-C}  & Zhao et al.~\cite{ZNI:MSCS:2003} applied on the  condensed graph  & $O(m+n)^{\dag}$  \\
 \hline
  \textsf{IST-B original}  & \textsf{2EC-B}  & Original sparse certificate from \cite{2ECB}  & $O(m+n)$  \\
  \textsf{IST-B}                   & \textsf{2EC-B} &  Modified sparse certificate & $O(m+n)$  \\
  \textsf{Test2EDP-B}             & \textsf{2EC-B} & Two edge-disjoint paths test on sparse certificate of input graph & $O(n^2)$ \\
  \textsf{Test2ECB-B}             & \textsf{2EC-B} & $2$-edge-connected blocks test on sparse certificate of input graph & $O(n^2)$ \\
  \textsf{Hybrid-B}                   & \textsf{2EC-B} & Hybrid of two edge-disjoint paths and $2$-edge-connected blocks test on sparse certificate of input graph & $O(n^2)$ \\
   \textsf{Test2EDP-B-Aux}             & \textsf{2EC-B} & \textsf{Test2EDP-B} applied on second-level auxiliary graphs & $O(n^2)$ \\
   \textsf{Hybrid-B-Aux}                   & \textsf{2EC-B} & \textsf{Hybrid-B} applied on second-level auxiliary graphs & $O(n^2)$ \\
  \hline
  \textsf{IST-BC}                   & \textsf{2EC-B-C} &  Modified sparse certificate preserving $2$-edge-connected components (applied on condensed graph) & $O(m+n)^{\dag}$  \\
  \textsf{Test2EDP-BC}             & \textsf{2EC-B-C} & Two edge-disjoint paths test on sparse certificate of condensed graph & $O(n^2)$ \\
   \textsf{Test2ECB-BC}             & \textsf{2EC-B-C} & $2$-edge-connected blocks test on sparse certificate of condensed graph & $O(n^2)$ \\
  \textsf{Hybrid-BC}                   & \textsf{2EC-B-C} & Hybrid of two edge-disjoint paths and $2$-edge-connected blocks test on sparse certificate of condensed graph & $O(n^2)$ \\
   \textsf{Test2EDP-BC-Aux}             & \textsf{2EC-B-C} & \textsf{Test2EDP-BC} applied on second-level auxiliary graphs & $O(n^2)$ \\
   \textsf{Hybrid-BC-Aux}                   & \textsf{2EC-B-C} & \textsf{Hybrid-BC} applied on second-level auxiliary graphs & $O(n^2)$ \\
  \hline
  \end{tabu}
\end{small}
\end{adjustbox}
\caption{The algorithms considered in our experimental study. The
worst-case bounds refer to a digraph with $n$ vertices and $m$
edges. $^{\dag}$These linear running times assume that the $2$-edge-connected components of the input digraph are available.}
\label{tab:algorithms}
\end{table}
\begin{table}[!ht]
\setlength{\tabcolsep}{4.5pt}
\begin{adjustbox}{width=\textwidth,totalheight=\textheight,keepaspectratio}
\begin{small}
\begin{tabular}{l|rrrrrrr|l}
\hline
Dataset & $n$ & $m$ & file size & $\delta_{avg}$  & $b^{\ast}$ & $\delta_{avg}^{\mathit{B}}$ & $\delta_{avg}^{\mathit{C}}$ &  type  \\
\hline
Rome99           & 3353       & 8859       & 100KB              & 2.64                    & 1474                    & 1.75                              & 1.67                                                  & road network                      \\
P2p-Gnutella25   & 5153       & 17695      & 203KB              & 3.43                    & 2181                    & 1.60                              & 1.00                                                  & peer2peer                         \\
P2p-Gnutella31   & 14149      & 50916      & 621KB              & 3.59                    & 6673                    & 1.56                              & 1.00                                                  & peer2peer                         \\
Web-NotreDame    & 53968      & 296228     & 3,9MB              & 5.48                    & 34879                   & 1.50                              & 1.36                                                  & web graph                         \\
Soc-Epinions1    & 32223      & 443506     & 5,3MB              & 13.76                   & 20975                   & 1.56                              & 1.55                                                  & social network                    \\
USA-road-NY      & 264346     & 733846     & 11MB               & 2.77                    & 104618                  & 1.80                              & 1.80                                                  & road network                      \\
USA-road-BAY     & 321270     & 800172     & 12MB               & 2.49                    & 196474                  & 1.69                              & 1.69                                                  & road network                      \\
USA-road-COL     & 435666     & 1057066    & 16MB               & 2.42                    & 276602                  & 1.68                              & 1.68                                                  & road network                      \\
Amazon0302       & 241761     & 1131217    & 16MB               & 4.67                    & 73361                   & 1.74                              & 1.64                                                  & prod. co-purchase                 \\
WikiTalk         & 111881     & 1477893    & 18MB               & 13.20                   & 85503                   & 1.45                              & 1.44                                                  & social network                    \\
Web-Stanford     & 150532     & 1576314    & 22MB               & 10.47                   & 64723                   & 1.62                              & 1.33                                                  & web graph                         \\
Amazon0601       & 395234     & 3301092    & 49MB               & 8.35                    & 83995                   & 1.82                              & 1.82                                                  & prod. co-purchase                 \\
Web-Google       & 434818     & 3419124    & 50MB               & 7.86                    & 211544                  & 1.59                              & 1.48                                                  & web graph                         \\
Web-Berkstan     & 334857     & 4523232    & 68MB               & 13.50                   & 164779                  & 1.56                              & 1.39                                           & web graph\\
\hline \end{tabular}

\end{small}
\end{adjustbox}
\caption{Real-world graphs sorted by file size of their largest
SCC; $n$ is the number of vertices, $m$ the number of edges, and
$\delta_{avg}$ is the average vertex indegree; $b^{\ast}$ is the
number of strong bridges; $\delta_{avg}^{B}$ and
$\delta_{avg}^{C}$ are lower bounds on the average vertex indegree
of an optimal solution to \textsf{2EC-B} and \textsf{2EC-C},
respectively.}\label{tab:datasets}
\end{table}

For the experimental evaluation we use the datasets shown in
Table~\ref{tab:datasets}. We measure the quality of the solution
computed by algorithm $A$ on problem ${\cal P}$ by a \emph{quality
ratio} defined as $q(A,{\cal P}) =
\delta_{\mathit{avg}}^{A}/\delta_{\mathit{avg}}^{{\cal P}}$, where
$\delta_{\mathit{avg}}^{A}$ is the average vertex indegree of the spanning
subgraph computed by $A$ and $\delta_{\mathit{avg}}^{{\cal P}}$ is
a lower bound on the average vertex indegree of the optimal
solution for ${\cal P}$.
Specifically, for \textsf{2EC-B} and
\textsf{2EC-B-C} we define
$\delta_{\mathit{avg}}^{B} = (n+k)/n$, where $n$ is the total number of vertices of the input digraph and $k$ is the number of vertices that belong in nontrivial $2$-edge-connected blocks~\footnote{This follows from the fact that in the sparse subgraph the $k$ vertices in nontrivial blocks must have indegree at least two, while the remaining $n-k$ vertices must have indegree at least one, since we seek for a strongly connected spanning subgraph.}.
%
We set a similar lower bound $\delta_{\mathit{avg}}^{C}$ for \textsf{2EC-C}, with the only difference that $k$ is the number of vertices that belong in nontrivial $2$-edge-connected components.
Note that the quality ratio is an upper bound of the actual approximation ratio of the specific input. 
The smaller the values of  $q(A,{\cal P})$ (i.e., the closer to 1), the better is the approximation obtained by algorithm $A$ for problem ${\cal P}$.

\subsection{Experimental results}
\label{sec:results}
We now report the results of our experiments with all the algorithms
considered for problems \textsf{2EC-B}, \textsf{2EC-B-C}  and \textsf{2EC-C}. 
As previously mentioned, for the sake of efficiency, all variants of \textsf{Test2EDP}, \textsf{Test2ECB} and \textsf{Hybrid}
were run on the sparse certificate computed by either \textsf{IST-B} or \textsf{IST-BC}
(depending on the problem at hand) instead of the original digraph.

We group the experimental results into two categories: results on the \textsf{2EC-B} problem and results on both \textsf{2EC-C} and \textsf{2EC-B-C} problems.
In all cases we are interested in the quality ratio of the computed solutions and the corresponding running times.
Moreover, in order to better highlight the different behaviour of our algorithms, we present for each algorithm both the quality ratio for each individual input and also give
an overall view in terms of box-and-whisker diagrams.
Specifically, we report the following experimental results:
\begin{itemize}
\item For the \textsf{2EC-B} problem:
\begin{itemize}
\item the quality ratio of the spanning  subgraphs computed by the different algorithms is shown in Table~\ref{tab:2ECB}, Figure~\ref{fig:plottedquality} (top), and Figure~\ref{fig:box-quality} (top);
\item their running times are given in Table~\ref{tab:2EC-B-times}, while the corresponding plotted values are shown in Figure~\ref{fig:2ecb-running-time} (top).
\end{itemize}
\item For the \textsf{2EC-C} and \textsf{2EC-B-C} problems:
\begin{itemize}
\item the quality ratio of the spanning  subgraphs computed by the different algorithms is shown in Table~\ref{tab:2ECBC}, Figure~\ref{fig:plottedquality} (bottom), and Figure~\ref{fig:box-quality} (bottom);
\item their running times are given in Table~\ref{tab:2EC-B-C-times}, while the corresponding plotted values are shown in Figure~\ref{fig:2ecb-running-time} (bottom). 
We note that the running times include the time to compute the $2$-edge-connected components of the input digraph. 
To that end, we use the algorithm from \cite{2ECB-Exp}, which is fast in practice despite the fact that its worst-case running time is $O(mn)$.
\end{itemize}
\end{itemize}


\begin{table}[!ht]
\setlength{\tabcolsep}{2.9pt}
\begin{adjustbox}{width=\textwidth,totalheight=\textheight,keepaspectratio}
\begin{small}
\begin{tabular}{l|ccccccc}
\multicolumn{1}{l|}{\multirow{2}{*}{\textsf{Dataset}}} & \multicolumn{1}{c}{\textsf{IST-B}} & \multicolumn{1}{c}{\multirow{2}{*}{\ \textsf{IST-B}}\ } & \multicolumn{1}{c}{\multirow{2}{*}{\ \textsf{Test2EDP-B}}\ }  & \multicolumn{1}{c}{\textsf{Test2ECB-B}} & \multicolumn{1}{l}{\multirow{2}{*}{\ \textsf{Test2EDP-B-Aux}}\ } & \multicolumn{1}{l}{\multirow{2}{*}{\ \textsf{Hybrid-B-Aux}}\ } \\
 & \multicolumn{1}{c}{original} & & & \multicolumn{1}{c}{\textsf{\& Hybrid-B}} & & \\
\hline
Rome99                              & 1.389                              & 1.363                                         & 1.171                                 & 1.167                                        & 1.177                                     & 1.174      \\
P2p-Gnutella25                      & 1.656                              & 1.512                                         & 1.220                                 & 1.143                                        & 1.251                                     & 1.234      \\
P2p-Gnutella31                      & 1.682                              & 1.541                                         & 1.251                                 & 1.169                                        & 1.291                                     & 1.274      \\
Web-NotreDame                       & 1.964                              & 1.807                                         & 1.489                                 & 1.417                                        & 1.500                                     & 1.471      \\
Soc-Epinions1                       & 2.047                              & 1.837                                         & 1.435                                 & 1.379                                        & 1.441                                     & 1.406      \\
USA-road-NY                         & 1.343                              & 1.245                                         & 1.174                                 & 1.174                                        & 1.175                                     & 1.175      \\
USA-road-BAY                        & 1.361                              & 1.307                                         & 1.245                                 & 1.246                                        & 1.246                                     & 1.246      \\
USA-road-COL                        & 1.354                              & 1.304                                         & 1.251                                 & 1.252                                        & 1.252                                     & 1.252      \\
Amazon0302                          & 1.762                              & 1.570                                         & 1.186                                 & 1.134                                        & 1.206                                     & 1.196      \\
WikiTalk                            & 2.181                              & 2.050                                         & 1.788                                 & 1.588                                        & 1.792                                     & 1.615      \\
Web-Stanford                        & 1.907                              & 1.688                                         & 1.409                                 & 1.365                                        & 1.418                                     & 1.406      \\
Amazon0601                          & 1.866                              & 1.649                                         & 1.163                                 & 1.146                                        & 1.170                                     & 1.166      \\
Web-Google                          & 1.921                              & 1.728                                         & 1.389                                 & 1.322                                        & 1.401                                     & 1.377          \\
Web-Berkstan                        & 2.048                              & 1.775                                         & 1.480                                 & 1.427                                        & 1.489                                     & 1.469         \\
\hline
\end{tabular} 
\end{small}
\end{adjustbox}
\caption{Quality ratio $q(A,{\cal P})$ of the solutions computed for \textsf{2EC-B}.}\label{tab:2ECB}

\vspace*{0.3in}

\setlength{\tabcolsep}{2.1pt}
\begin{adjustbox}{width=\textwidth,totalheight=\textheight,keepaspectratio}
\begin{small}
\begin{tabular}{l|c|cccccc}
	\multicolumn{1}{l|}{\multirow{2}{*}{\textsf{Dataset}}} & \multicolumn{1}{c|}{\multirow{2}{*}{\textsf{ZNI-C}}} & \multicolumn{1}{c}{\multirow{2}{*}{\textsf{IST-BC}}} & \multicolumn{1}{c}{\multirow{2}{*}{\textsf{Test2EDP-BC}}} & \multicolumn{1}{c}{\multirow{1}{*}{\textsf{Test2ECB-BC}}}& \multicolumn{1}{l}{\textsf{\multirow{2}{*}{Test2EDP-BC-Aux}}} & \multicolumn{1}{l}{\textsf{\multirow{2}{*}{Hybrid-BC-Aux}}} \\  & & & &\multicolumn{1}{c}{\textsf{\& Hybrid-BC}} & & \\\hline
Rome99                              & 1.360                                               & 1.371                                                  & 1.197                                  & 1.187                                  & 1.197                                      & 1.195                                      \\
P2p-Gnutella25                      & 1.276                                               & 1.517                                                  & 1.218                                  & 1.141                                  & 1.249                                      & 1.232                                      \\
P2p-Gnutella31                      & 1.312                                               & 1.537                                                  & 1.251                                  & 1.170                                  & 1.290                                      & 1.273                                      \\
Web-NotreDame                       & 1.620                                               & 1.747                                                  & 1.500                                  & 1.426                                  & 1.510                                      & 1.484                                      \\
Soc-Epinions1                       & 1.790                                               & 1.847                                                  & 1.488                                  & 1.435                                  & 1.489                                      & 1.476                                      \\
USA-road-NY                         & 1.343                                               & 1.341                                                  & 1.163                                  & 1.163                                  & 1.163                                      & 1.163                                      \\
USA-road-BAY                        & 1.360                                               & 1.357                                                  & 1.237                                  & 1.237                                  & 1.237                                      & 1.237                                      \\
USA-road-COL                        & 1.343                                               & 1.339                                                  & 1.242                                  & 1.242                                  & 1.242                                      & 1.242                                      \\
Amazon0302                          & 1.464                                               & 1.580                                                  & 1.279                                  & 1.228                                  & 1.292                                      & 1.284                                      \\
WikiTalk                            & 1.891                                               & 2.099                                                  & 1.837                                  & 1.630                                  & 1.838                                      & 1.827                                      \\
Web-Stanford                        & 1.560                                               & 1.679                                                  & 1.430                                  & 1.390                                  & 1.436                                      & 1.427                                      \\
Amazon0601                          & 1.709                                               & 1.727                                                  & 1.200                                  & 1.186                                  & 1.202                                      & 1.200                                      \\
Web-Google                          & 1.637                                               & 1.728                                                  & 1.437                                  & 1.381                                  & 1.446                                      & 1.431                                      \\
Web-Berkstan                        & 1.637                                               & 1.753                                                  & 1.516                                  & 1.472                                  & 1.523                                      & 1.511  \\
	\hline
\end{tabular} 
\end{small}
\end{adjustbox}
\caption{Quality ratio $q(A,{\cal P})$ of the solutions computed for \textsf{2EC-C} and \textsf{2EC-B-C}.}\label{tab:2ECBC}
\end{table}


\begin{figure}[!ht]
\centering
\includegraphics[width=0.9\textwidth]{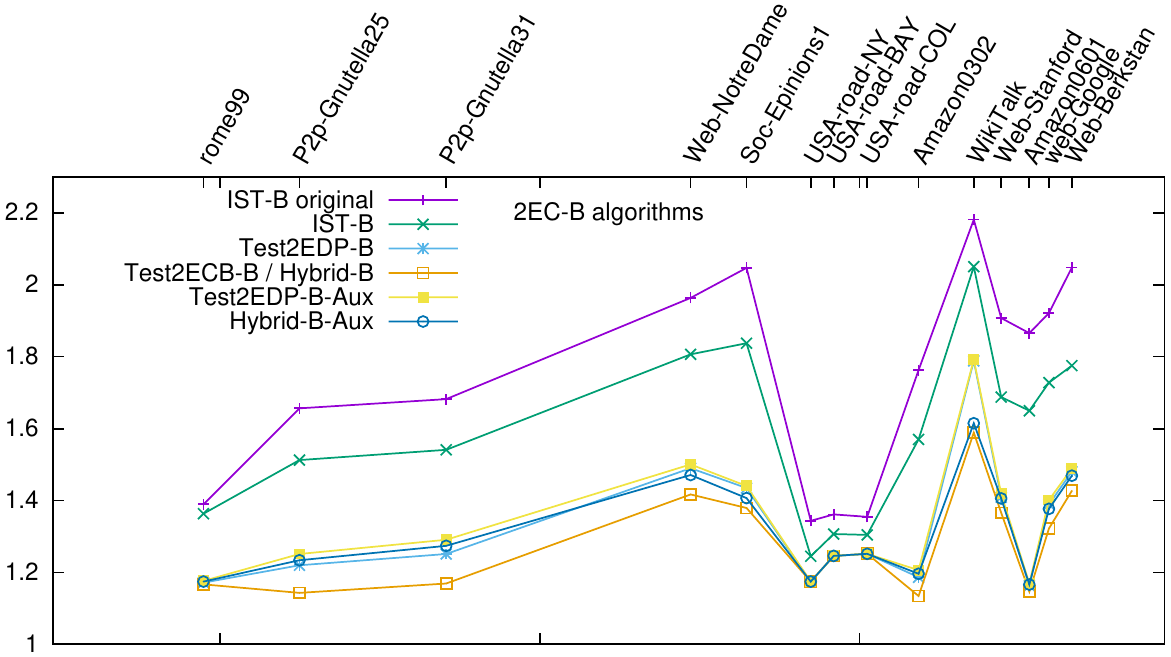}%

\vspace*{0.1in}

\includegraphics[width=0.9\textwidth]{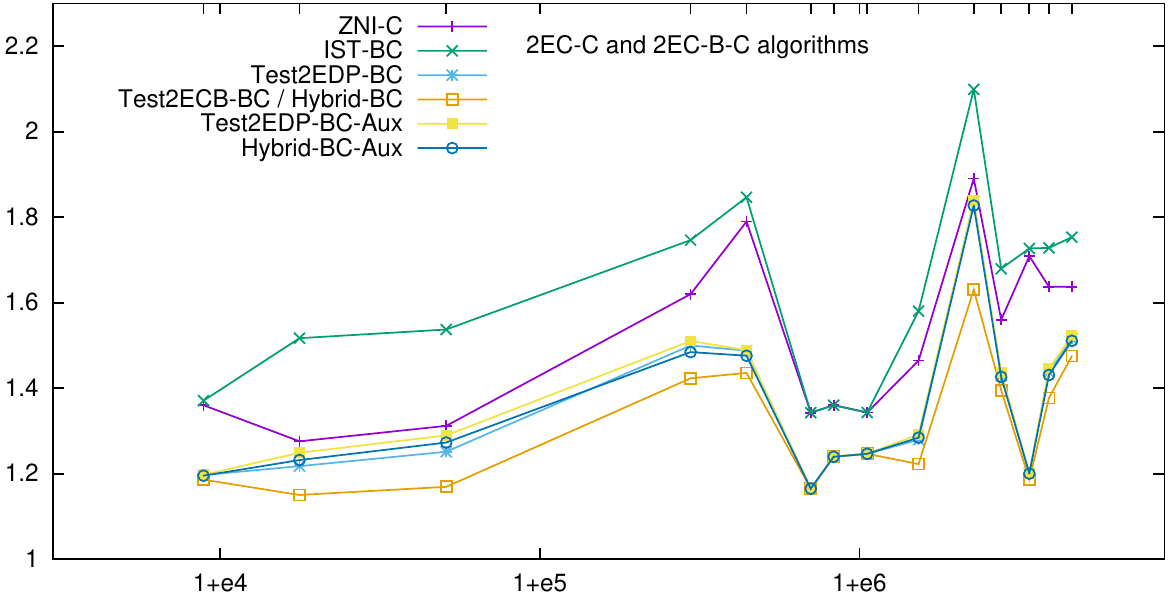}%
\caption{The plotted quality ratios taken from Tables~\ref{tab:2ECB} and \ref{tab:2ECBC}, respectively.
\label{fig:plottedquality}}
\end{figure}


\begin{figure}[!ht]
\centering
\includegraphics[width=0.8\textwidth]{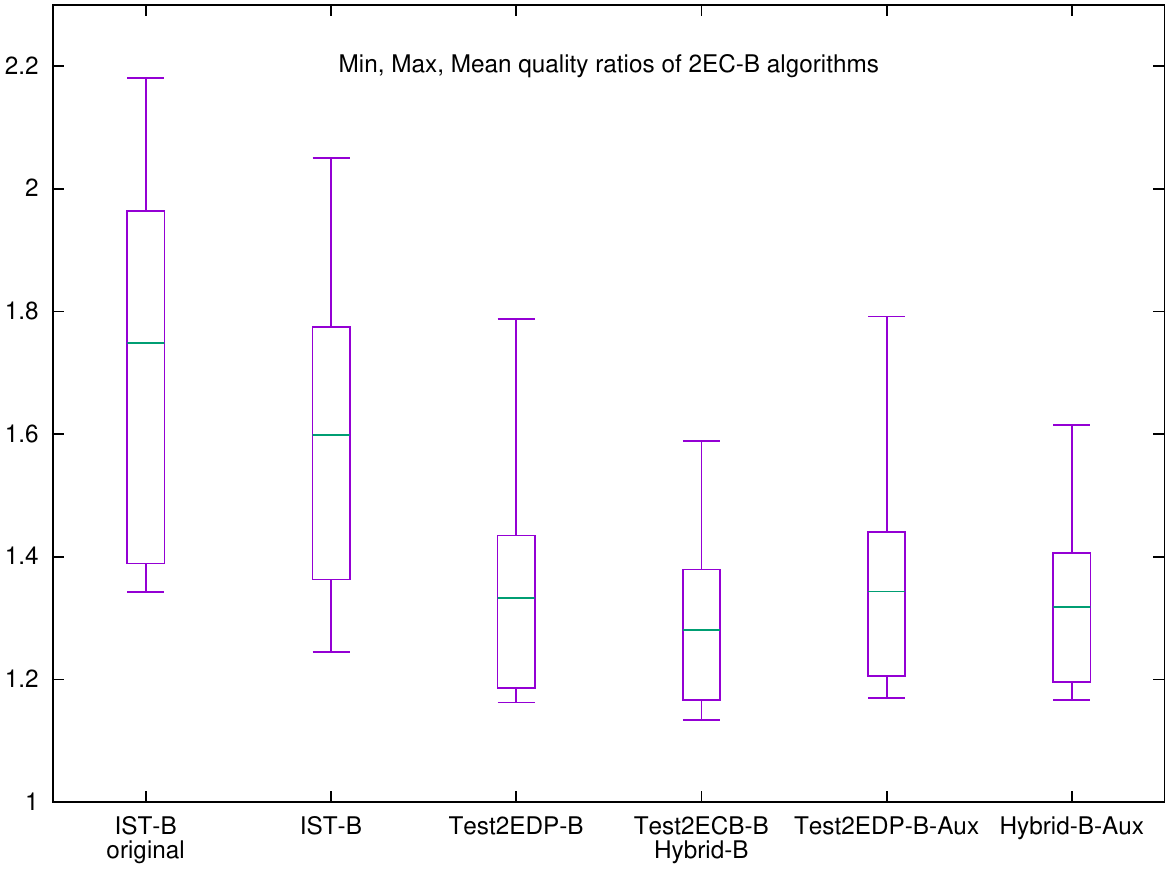}%

\vspace*{0.1in}

\includegraphics[width=0.8\textwidth]{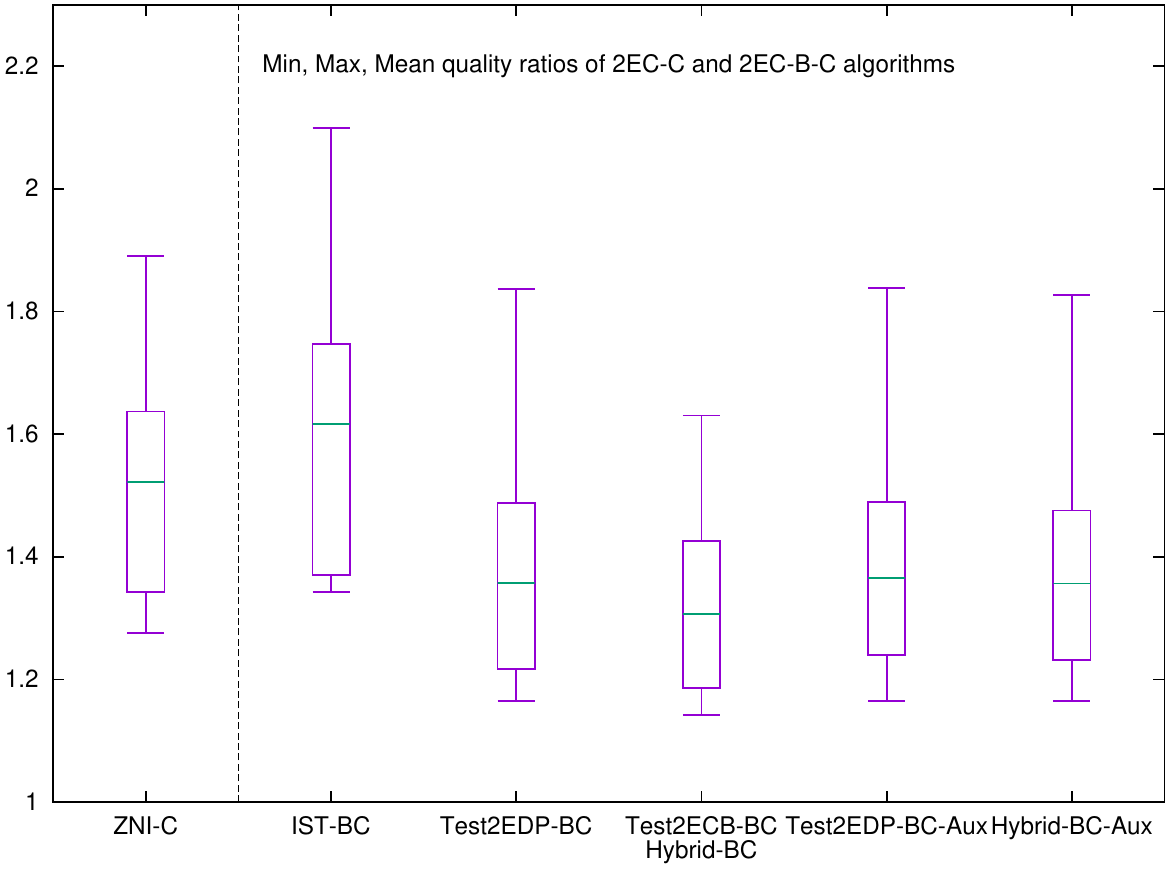}%
\caption{The quality ratios in terms of box-and-whisker diagrams. The range of each box is obtained from half of the datasets.
\label{fig:box-quality}}
\end{figure}


\begin{table}[!ht]
\setlength{\tabcolsep}{3.3pt}
\begin{adjustbox}{width=\textwidth,totalheight=\textheight,keepaspectratio}
\begin{scriptsize}
\begin{tabular}{l|d{1.3}d{1.3}d{4.3}d{6.3}d{5.3}d{3.3}d{5.3}}
\multicolumn{1}{l|}{\multirow{2}{*}{\textsf{Dataset}}} &
\multicolumn{1}{c}{\textsf{IST-B}} &
\multicolumn{1}{c}{\multirow{2}{*}{\ \textsf{IST-B}}\ } &
\multicolumn{1}{c}{\multirow{2}{*}{\ \textsf{Test2EDP-B}}\ }  &
\multicolumn{1}{c}{\multirow{2}{*}{\textsf{Test2ECB-B}}} &
\multicolumn{1}{c}{\multirow{2}{*}{\textsf{Hybrid-B}}} &
\multicolumn{1}{l}{\multirow{2}{*}{\ \textsf{Test2EDP-B-Aux}}\ } &
\multicolumn{1}{l}{\multirow{2}{*}{\ \textsf{Hybrid-B-Aux}}\ } \\
 & \multicolumn{1}{c}{original} & & & & & \\
\hline
rome99         & 0.008 & 0.010 & 0.160    & 14.297      & 0.224     & 0.056   & 0.183     \\
P2p-Gnutella25 & 0.017 & 0.019 & 0.848    & 44.377      & 3.767     & 0.295   & 2.595     \\
P2p-Gnutella31 & 0.052 & 0.064 & 6.716    & 352.871     & 31.935    & 2.467   & 20.923    \\
Web-NotreDame  & 0.211 & 0.281 & 46.937   & 4723.904    & 352.834   & 3.192   & 215.492   \\
Soc-Epinions1  & 0.194 & 0.224 & 47.869   & 2073.662    & 135.098   & 16.387  & 234.066   \\
USA-road-NY    & 0.648 & 0.788 & 750.874  & 81990.402   & 206.055   & 110.616 & 108.463   \\
USA-road-BAY   & 0.979 & 1.212 & 1002.689 & 132171.251  & 475.378   & 186.816 & 187.277   \\
USA-road-COL   & 1.333 & 1.681 & 1794.103 & 231785.495  & 976.019   & 217.215 & 214.586   \\
Amazon0302     & 1.068 & 1.253 & 1398.438 & 164047.057  & 8499.349  & 331.569 & 3985.706  \\
WikiTalk       & 0.763 & 0.918 & 637.879  & 28339.485   & 5057.806  & 91.674  & 10877.771 \\
Web-Stanford   & 0.908 & 1.309 & 607.356  & 49532.517   & 2120.636  & 25.184  & 952.585   \\
Amazon0601     & 2.406 & 2.698 & 4847.592 & 446475.698  & 8408.463  & 968.964 & 8382.981  \\
web-Google     & 3.362 & 3.898 & 4801.787 & 612329.017  & 38031.588 & 422.058 & 25899.907 \\
Web-Berkstan   & 1.829 & 3.841 & 2180.488 & 212587.201  & 10805.487 & 96.372  & 5641.406 \\
\hline
\end{tabular} 
\end{scriptsize}
\end{adjustbox}
\caption{Running times in seconds of the algorithms for the \textsf{2EC-B} problem.\label{tab:2EC-B-times}}

\vspace*{0.3in}

\setlength{\tabcolsep}{3.3pt}
\begin{adjustbox}{width=\textwidth,totalheight=\textheight,keepaspectratio}
\begin{scriptsize}
\begin{tabular}{l|d{1.3}|d{2.3}d{4.3}d{6.3}d{5.3}d{4.3}d{5.3}}
\multicolumn{1}{l|}{{\textsf{Dataset}}} &
\multicolumn{1}{c|}{{\textsf{ZNI-C}}} &
\multicolumn{1}{c}{{\textsf{IST-BC}}} &
\multicolumn{1}{c}{{\textsf{Test2EDP-BC}}} &
\multicolumn{1}{c}{{\textsf{Test2ECB-BC}}}&
\multicolumn{1}{c}{{\textsf{Hybrid-BC}}} &
\multicolumn{1}{l}{\textsf{Test2EDP-BC-Aux}} &
\multicolumn{1}{l}{\textsf{Hybrid-BC-Aux}} \\
\hline
rome99         & 0.012 & 0.019  & 0.051    & 1.013      & 0.126     & 0.054    & 0.154     \\
P2p-Gnutella25 & 0.010 & 0.029  & 0.855    & 77.274     & 3.727     & 0.320    & 2.574     \\
P2p-Gnutella31 & 0.025 & 0.090  & 6.438    & 664.936    & 31.348    & 2.495    & 20.644    \\
Web-NotreDame  & 0.159 & 0.448  & 11.062   & 2635.104   & 267.482   & 2.036    & 165.532   \\
Soc-Epinions1  & 0.177 & 0.442  & 10.778   & 203.688    & 61.531    & 10.022   & 36.404    \\
USA-road-NY    & 0.339 & 2.000  & 208.987  & 244.003    & 214.563   & 209.334  & 209.309   \\
USA-road-BAY   & 0.437 & 4.539  & 151.786  & 289.465    & 178.197   & 152.488  & 152.407   \\
USA-road-COL   & 0.547 & 5.275  & 198.795  & 526.362    & 305.525   & 199.768  & 199.711   \\
Amazon0302     & 1.687 & 3.671  & 237.584  & 38201.229  & 3184.360  & 148.871  & 1909.122  \\
WikiTalk       & 0.923 & 6.182  & 131.766  & 3538.042   & 2620.733  & 66.261   & 407.962   \\
Web-Stanford   & 1.290 & 2.499  & 226.669  & 50153.480  & 1250.210  & 20.134   & 636.641   \\
Amazon0601     & 4.768 & 7.659  & 1732.197 & 13067.429  & 2791.030  & 1725.333 & 2390.567  \\
web-Google     & 6.275 & 18.988 & 892.954  & 204990.718 & 15783.304 & 345.384  & 11714.605 \\
Web-Berkstan   & 1.911 & 9.744  & 456.082  & 186129.463 & 5792.903  & 70.600   & 2552.911 \\
\hline
\end{tabular} 
\end{scriptsize}
\end{adjustbox}
\caption{Running times in seconds of the algorithms for the \textsf{2EC-C} and \textsf{2EC-B-C} problems.\label{tab:2EC-B-C-times}}
\end{table}


\begin{figure}[!ht]
\centering
\includegraphics[width=0.9\textwidth]{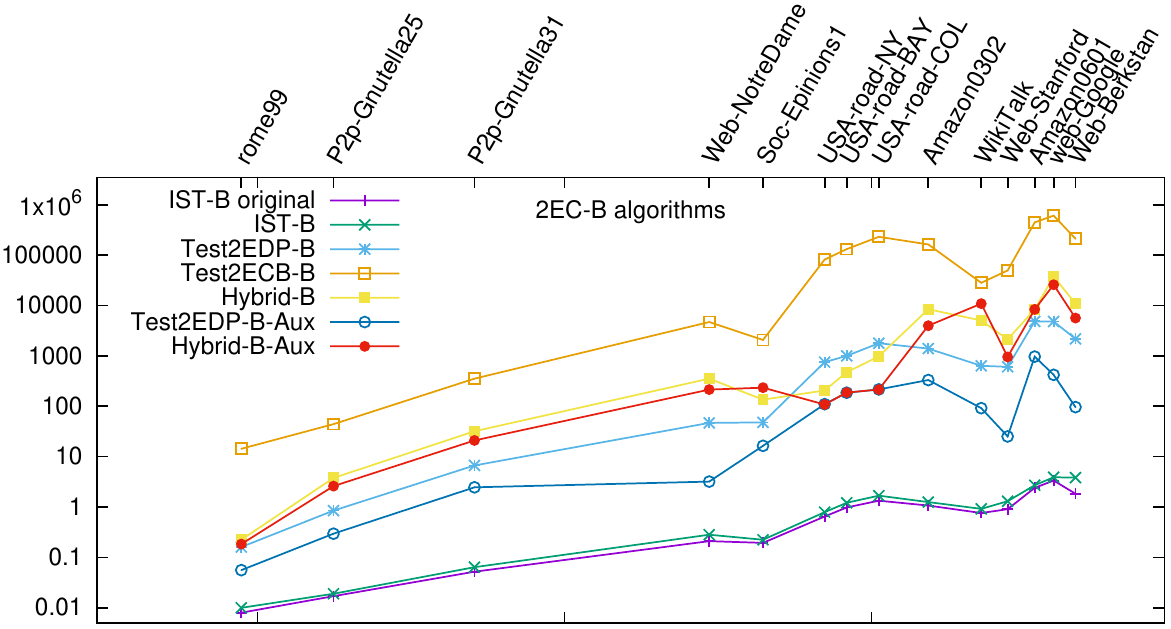}%

\vspace*{0.1in}

\includegraphics[width=0.9\textwidth]{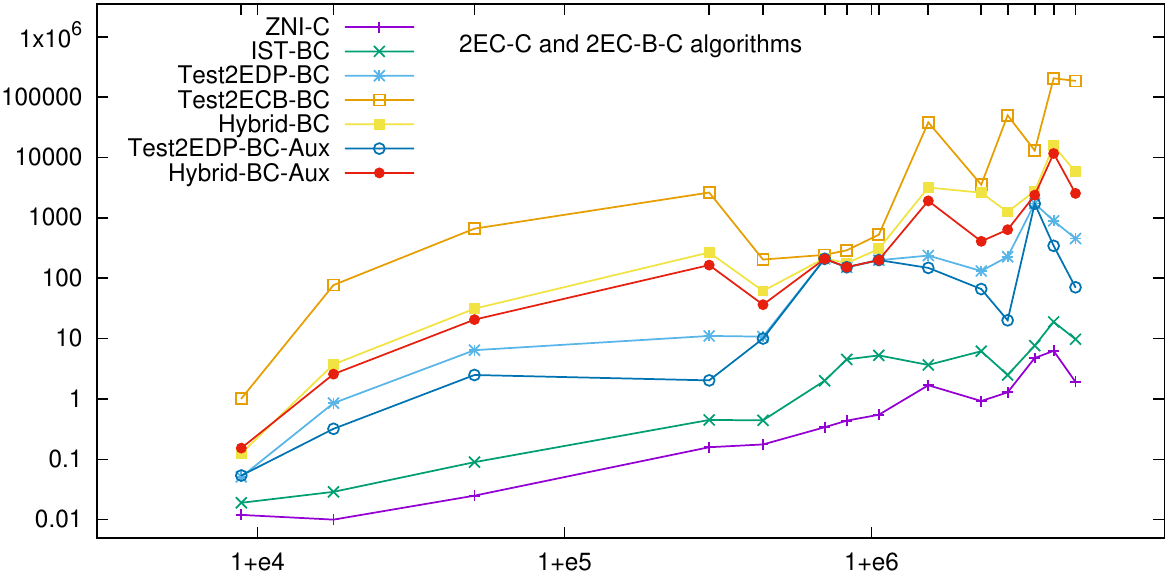}%
\caption{Running times in seconds with respect to the number of edges (in log-log scale).
}\label{fig:2ecb-running-time}
\end{figure}

\subsection{Evaluation of the experimental results}
\label{sec:discussion}
There are two peculiarities related to road networks that emerge immediately from the analysis of our experimental data.
First, all algorithms achieve consistently better approximations for road networks than for most of the other graphs in our data set.
Second, for the \textsf{2EC-B} problem the \textsf{Hybrid} algorithms (\textsf{Hybrid-B} and \textsf{Hybrid-B-Aux}) seem to achieve substantial speedups on road networks; for the
\textsf{2EC-B-C} problem, this is even true for
\textsf{Test2ECB-BC}.
The first phenomenon can be explained by taking into account the macroscopic structure of road networks, which is rather different from other networks.
Indeed, road networks are very close to be ``undirected": i.e., whenever there is an edge $(x,y)$, there is also the reverse edge $(y,x)$ (expect for one-way roads).
Roughly speaking, road networks mainly consist of the union of $2$-edge-connected components,  joined together by strong bridges, and their $2$-edge-connected blocks coincide with their $2$-edge-connected components.
In this setting, a sparse strongly connected subgraph of the condensed graph will preserve both blocks and components.
The second phenomenon is mainly due to the \emph{trivial edge} heuristic described in Section~\ref{sec:add-speed-up}.
%
%

Apart from the peculiarities of road networks, \textsf{ZNI-C} behaves as expected for \textsf{2EC-C} through its linear-time $2$-approximation algorithm.
Note that for both problems \textsf{2EC-B} and \textsf{2EC-B-C}, all algorithms achieve quality ratio significantly smaller than our theoretical bound of $4$.
Regarding running times, we observe that the \textsf{2EC-B-C} algorithms are faster than the \textsf{2EC-B} algorithms, sometimes significantly, as
they take advantage of the condensed graph that seems to admit small size in real-world applications. 
In addition, our experiments highlight interesting tradeoffs between practical performance and quality of the obtained solutions.
Indeed, the fastest (\textsf{IST-B} and \textsf{IST-B original} for problem \textsf{2EC-B};
\textsf{IST-BC} for \textsf{2EC-B-C})
and the slowest algorithms (\textsf{Test2ECB-B} and \textsf{Hybrid-B} for \textsf{2EC-B};
\textsf{Test2ECB-BC} and \textsf{Hybrid-BC} for \textsf{2EC-B-C}) tend to produce respectively the worst and the best approximations.
Note that
\textsf{IST-B} improves the quality of the solution of \textsf{IST-B original} at the price of slightly higher running times, while \textsf{Hybrid-B} (resp., \textsf{Hybrid-BC}) produces the same solutions as \textsf{Test2ECB-B} (resp., \textsf{Test2ECB-BC}) with rather impressive speedups. Running an algorithm on the second-level auxiliary graphs seems to produce substantial performance benefits at the price of a slightly worse approximation (\textsf{Test2EDP-B-Aux}, \textsf{Hybrid-B-Aux}, \textsf{Test2EDP-BC-Aux} and \textsf{Hybrid-BC-Aux}  versus \textsf{Test2EDP-B}, \textsf{Hybrid-B}, \textsf{Test2EDP-BC} and \textsf{Hybrid-BC}).
Overall,
in our experiments \textsf{Test2EDP-B-Aux}
and \textsf{Test2EDP-BC-Aux} seem to provide
good quality solutions for the problems considered without being  penalized too much by a substantial performance degradation.

\section{Concluding remarks}
\label{sec:concluding}

We do not know if the approximation ratio of $4$ that we provided for the algorithms of the \textsf{IST} and \textsf{Hybrid} families are tight. Figure \ref{fig:smallest-minimal}(a) shows a digraph $G$ such that a sparse certificate constructible by algorithm \textsf{IST-B} has $6n+O(1)$ edges. This digraph has a single nontrivial $2$-edge-connected block consisting of the vertices $x_1, x_2, \ldots, x_k$, which also form a $2$-edge-connected component.
An optimal solution for \textsf{2EC-B} on this instance, shown in Figure \ref{fig:smallest-minimal}(b), has $2n+O(1)$ edges, where each vertex $x_i$ has indegree and outdegree equal to two, while the other four vertices have indegree and oudegree equal to one. Figure \ref{fig:smallest-minimal}(c) shows a minimal solution with $3n+O(1)$ edges, where again each vertex $x_i$ has indegree and outdegree equal to two but vertex $y$ has indegree equal to $k$ and vertex $z$ has outdegree equal to $k$; removing any edge of this minimal solution either destroys the strong connectivity of the subgraph or partitions the nontrivial block.  So, for this instance \textsf{IST-B} achieves a $3$-approximation, while \textsf{Hybrid-B} achieves a $3/2$-approximation.
The three phases of the sparse certificate construction by \textsf{IST-B} are given in Figures \ref{fig:smallest-minimal-2}, \ref{fig:smallest-minimal-3} and \ref{fig:smallest-minimal-4}.

We also note that the example of Figure \ref{fig:smallest-minimal} is not a worst-case instance for \textsf{Hybrid-B}. If the input digraph is $2$-edge-connected then we seek for a smallest $2$-edge-connected spanning subgraph, and Lemma \ref{lemma:Test2EDP-Test2ECB} implies that \textsf{Hybrid-B} produces the same output as \textsf{Test2EDP-B}. So, in this case \textsf{Hybrid-B} achieves an approximation ratio of $2$, which is known to be tight~\cite{CT00}. 
In light of our experimental results, it seems possible that the \textsf{Hybrid} algorithms always achieve a $2$-approximation, but we have no proof.

\begin{figure}[t!]
\begin{center}
\includegraphics[trim={0 0 0 7cm}, clip=true, width=1.0\textwidth]{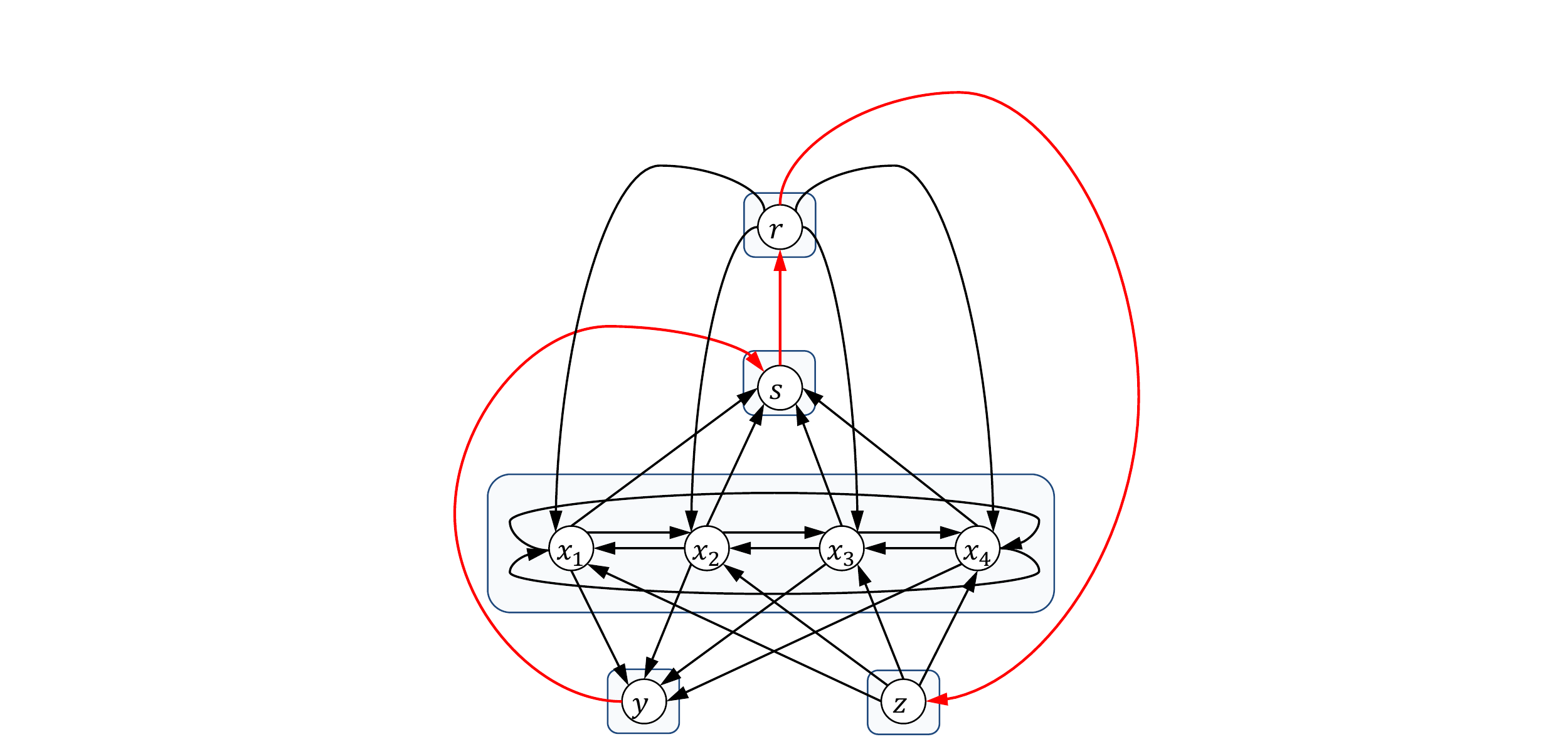}
(a)
\includegraphics[trim={0 0 0 7cm}, clip=true, width=1.0\textwidth]{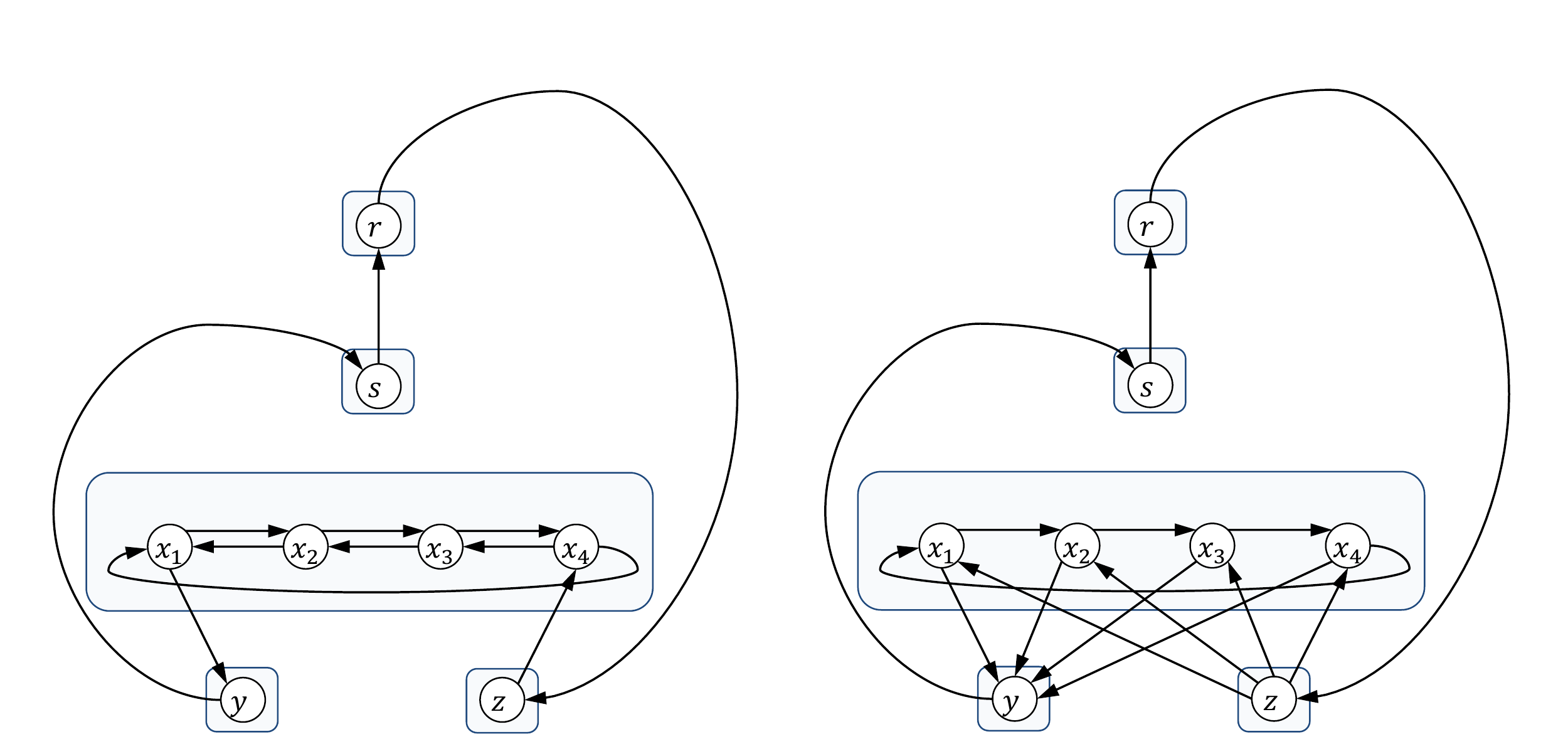}
(b) \hspace{7.5cm} (c) \hspace{2.5cm}
\end{center}
\caption{(a) A digraph $G$ with $n=k+4$ vertices and $m = 6n - 21$ edges (in this instance $k=4$). Strong bridges are shown in red (better viewed in color). Digraph $G$ has a single nontrivial  $2$-edge-connected block consisting of the vertices $x_1, x_2, \ldots, x_k$. (b) A minimum solution for the \textsf{2EC-B} problem with $2n-4$ edges. (c) A minimal solution for the \textsf{2EC-B} problem with $3n-9$ edges.}
\label{fig:smallest-minimal}
\end{figure}

\begin{figure}[h!]
\begin{center}
\includegraphics[trim={0 0 0 7cm}, clip=true, width=1.0\textwidth]{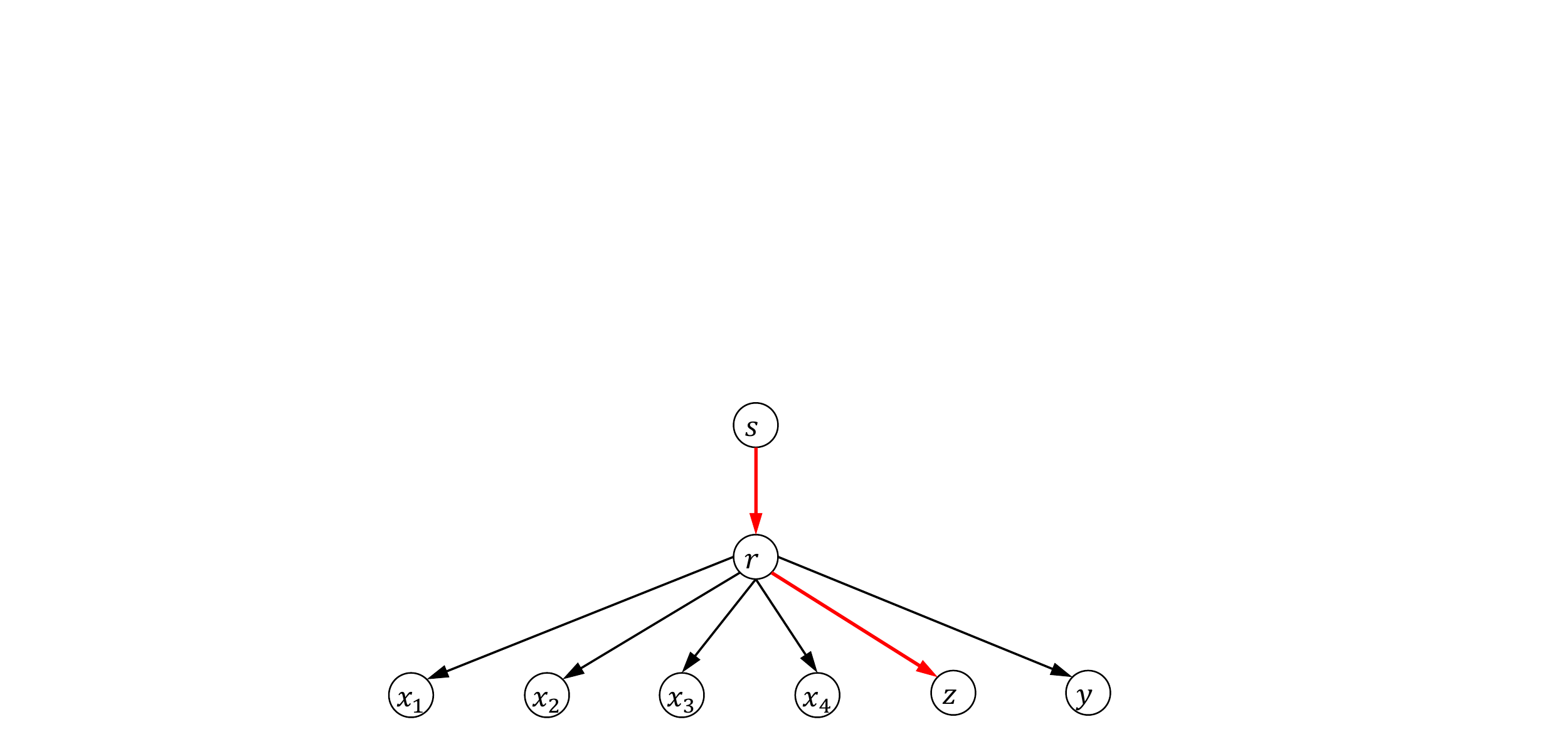}
\end{center}
\hspace{7.5cm} (a)
\begin{center}
\includegraphics[trim={0 0 0 7cm}, clip=true, width=1.0\textwidth]{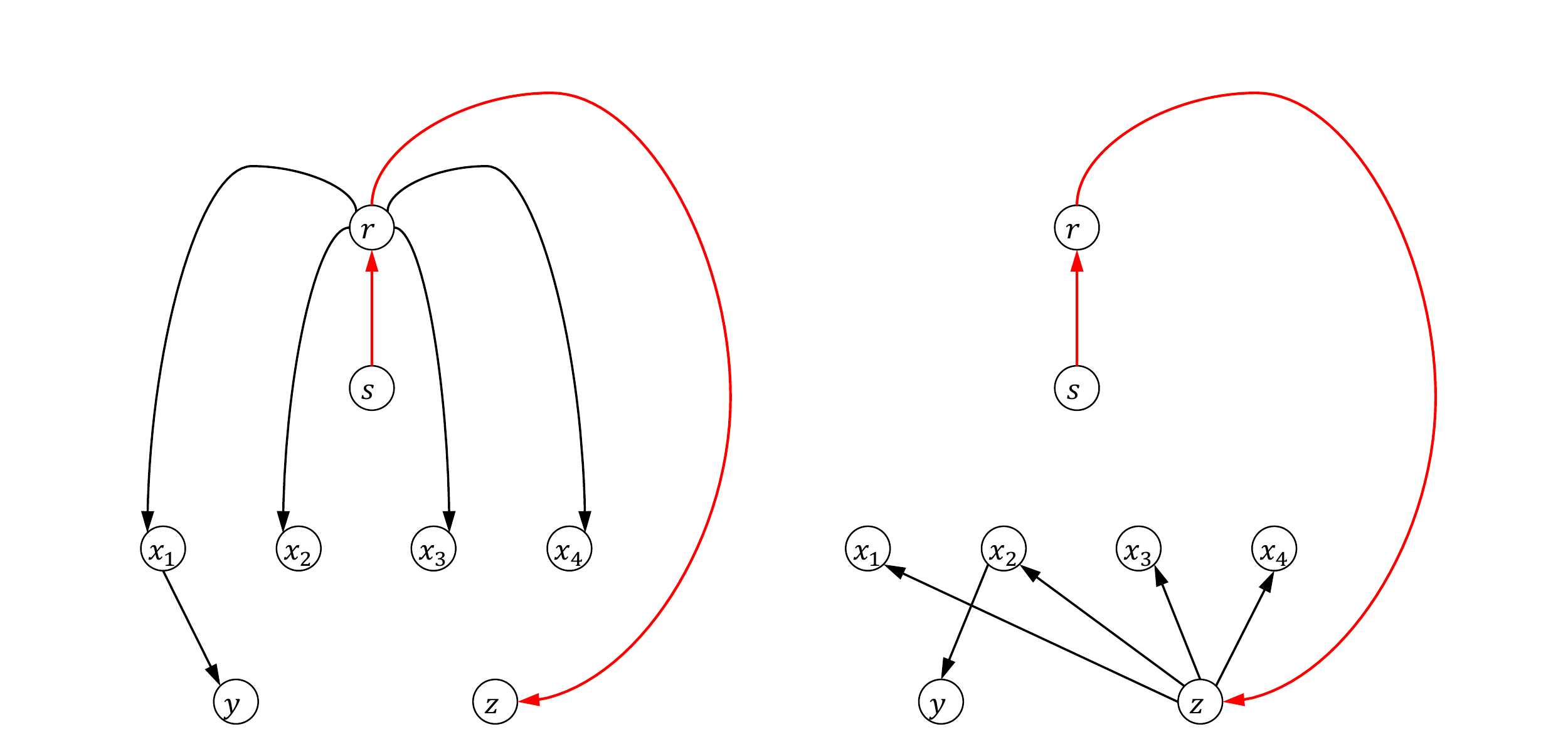}
\end{center}
\hspace{3.5cm} (b) \hspace{6.8cm} (c)
\caption{(a)  The dominator tree of the flow graph $G(s)$ that corresponds to digraph $G$ of Figure \ref{fig:smallest-minimal}. (b) and (c) Two independent spanning trees of $G(s)$ that may be selected by Phase 1 of the sparse certificate construction.}
\label{fig:smallest-minimal-2}
\end{figure}

\begin{figure}[h!]
\begin{center}
\includegraphics[trim={0 0 0 7cm}, clip=true, width=1.0\textwidth]{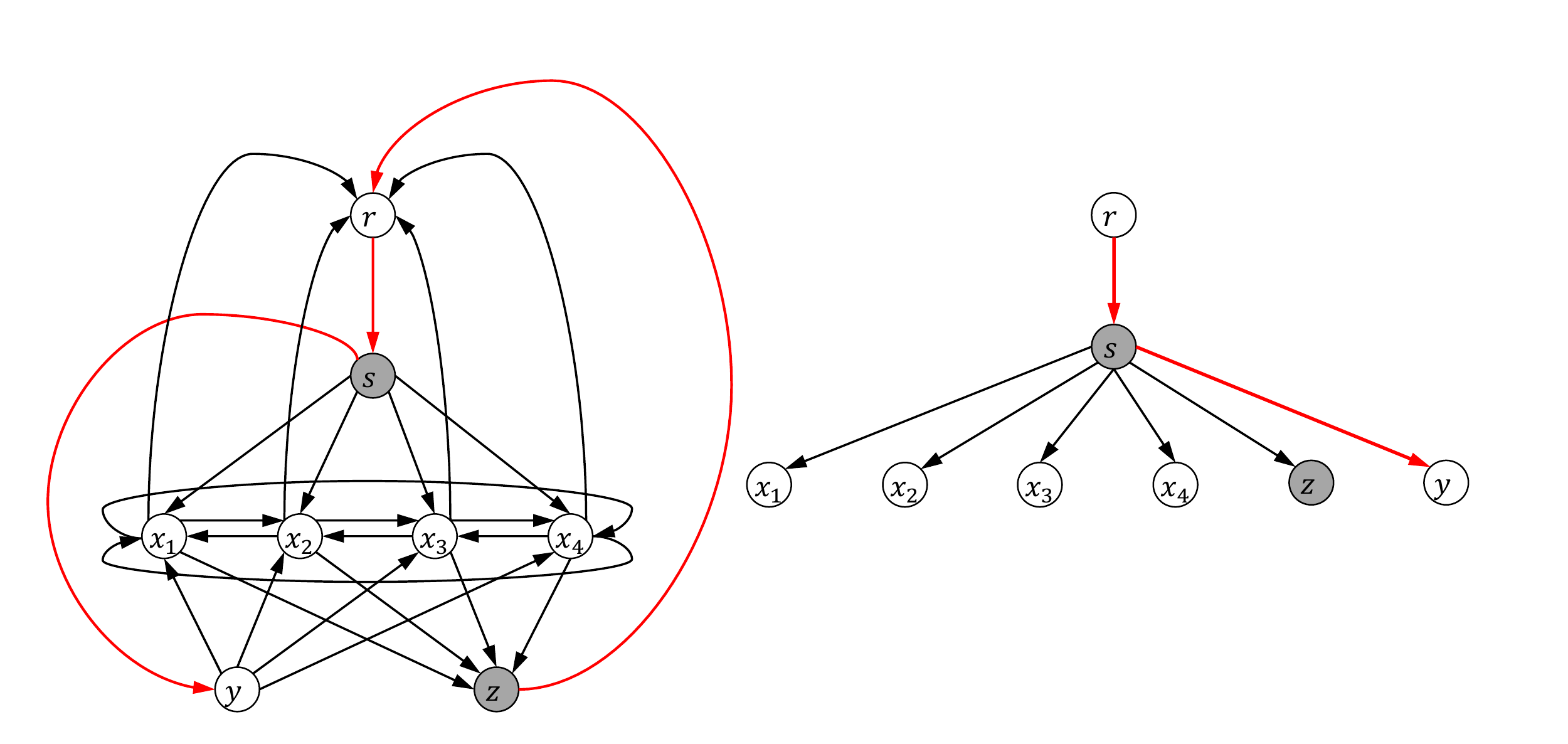}
\end{center}
\vspace{-0.2cm}
\hspace{4cm} (a) \hspace{7cm} (b)
\vspace{0.2cm}
\begin{center}
\includegraphics[trim={0 0 0 10cm}, clip=true, width=1.0\textwidth]{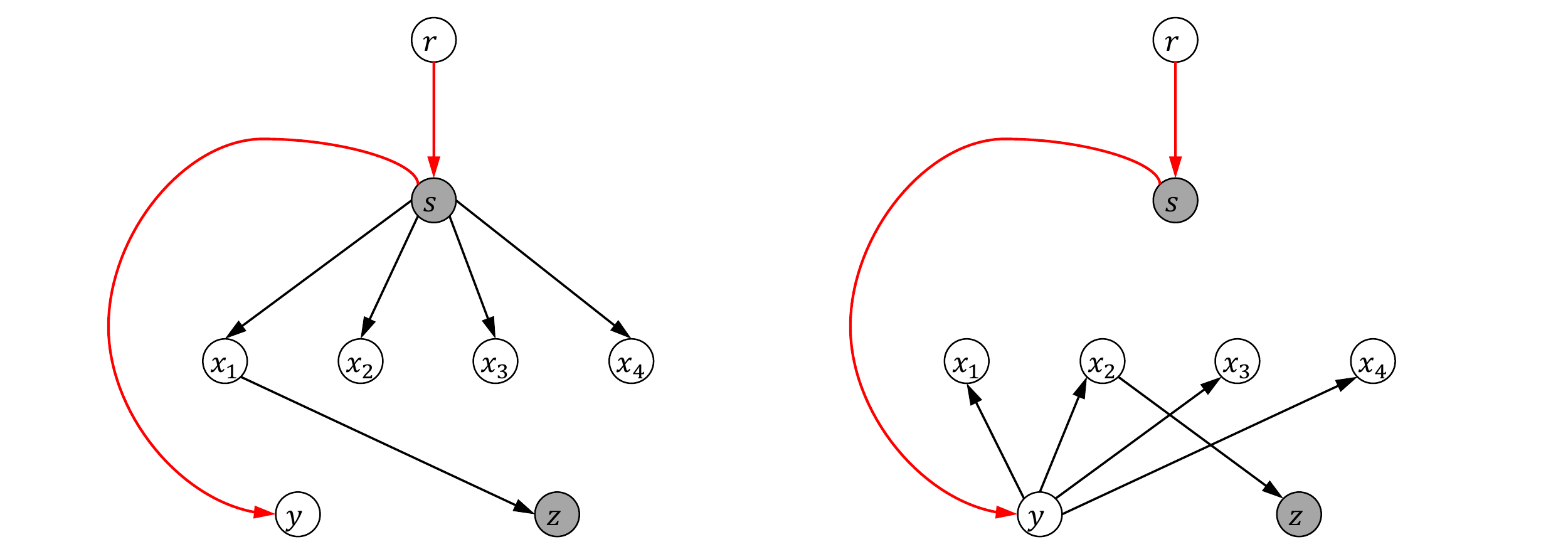}
\end{center}
\hspace{4cm} (c) \hspace{7cm} (d)
\vspace{-0.2cm}
\caption{(a) The reverse graph $H^R$ of the auxiliary graph $H=G_r$ of Figure \ref{fig:smallest-minimal-2}. Auxiliary vertices are shown grey. (b) The dominator tree of $H^R(r)$ with start vertex $r$.
 (c) and (d)  Two independent spanning trees of $H^R(r)$ that may be selected by Phase 2 of the sparse certificate construction.
}
\label{fig:smallest-minimal-3}
\end{figure}

\begin{figure}[h!]
\begin{center}
\includegraphics[trim={0 0 0 7cm}, clip=true, width=1.0\textwidth]{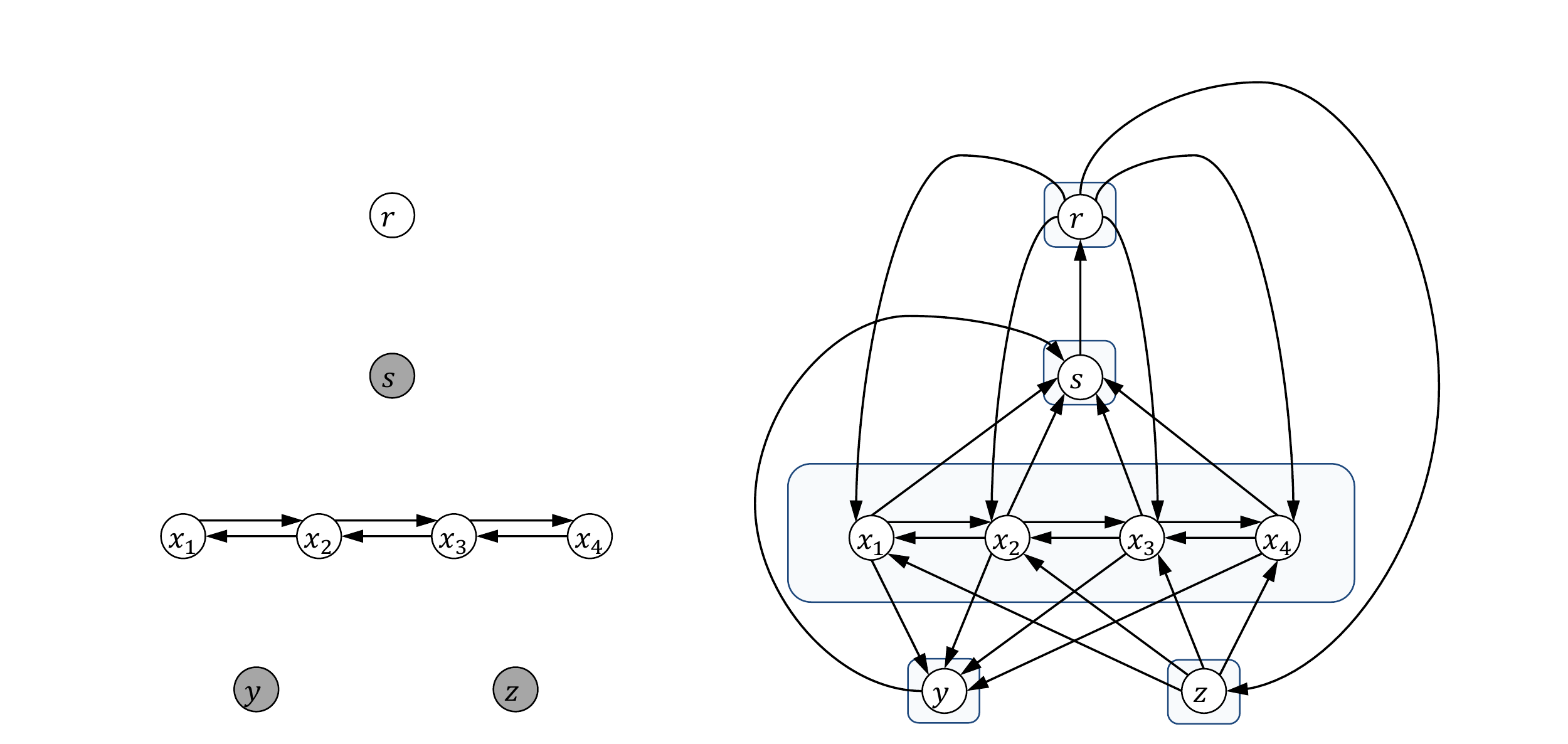}
\end{center}
\hspace{4cm} (a) \hspace{6.2cm} (b)
\caption{(a) The strongly connected components of $H^R_s \setminus (r,s)$, where $H^R_s$ is the second-level auxiliary graph of $H^R$ of Figure \ref{fig:smallest-minimal-3}. The only nontrivial component is induced by the vertices $x_1, x_2, \ldots, x_k$. The edges shown may be  selected by Phase 3 of the sparse certificate construction. (b) The final sparse certificate with $6n-23$ edges.}
\label{fig:smallest-minimal-4}
\end{figure}

We close with a couple of few more open questions and possible directions for future work.
First, we can consider the case of vertex-connectivity, where we can define the corresponding problems of computing the smallest strongly connected spanning subgraph that maintains the $2$-vertex-connected blocks, or
the $2$-vertex-connected components, or both.
A sparse certificate for the $2$-vertex-connected blocks is given in \cite{2VCB}, so it would interesting to study if based on this construction we can achieve a similar approximation ratio for the $2$-vertex-connectivity case.
Furthermore, the concept of $2$-edge-connected blocks may well be generalized to $k$-edge disjoint paths, for $k \geq 2$.
Keeping in mind that the underlying graph should remain strongly connected, it is natural to ask if computing smallest such spanning subgraph achieves a better approximation ratio for $k > 2$.
Such a phenomenon occurs in approximating the smallest spanning $k$-edge connected subgraph \cite{CT00,GGTW:kECSS:2009}.

\bibliographystyle{plain}
\bibliography{ltg}

\end{document}